\documentclass[11pt,onecolumn]{IEEEtran}
\usepackage{amsmath,amssymb,latexsym,cite}
\usepackage{pstricks,graphicx,pst-plot,pstricks-add}

\newtheorem{theorem}{Theorem}

\newtheorem{lemma}{Lemma}

\newcommand{\typ}[2]{\ensuremath{T_{#1}(#2)}}
\newcommand{\styp}[1]{\ensuremath{T_{#1}}}
\def\chunkeps{\epsilon_{4}}
\def\traineps{\epsilon_{5}}
\def\chunkmoveeps{\epsilon_{6}}
\def\cbookeps{\epsilon_{7}}
\def\rlesseps{\epsilon_{8}}
\def\ratelesserr{\hat{\varepsilon}}
\def\codechunk{c}

\def\schemeerr{\varepsilon_{\mathrm{dec}}}
\def\rateerr{\varepsilon_{\mathrm{ach}}}

\def\hb{\ensuremath{h}}

\def\rbits{k}
\def\chunk{b}
\def\train{t}
\def\badrate{\tau}
\def\algerror{\varepsilon}
\def\mlo{M_{\ast}}
\def\mhi{M^{\ast}}

\def\itrain{T}
\def\icode{U}

\def\Rfb{R_{\mathrm{fb}}}
\def\nfb{n_{\mathrm{fb}}}
\def\Bfb{B_{\mathrm{fb}}}

\newcommand{\empchanshort}[1]{\ensuremath{W_{#1}}}
\newcommand{\empchan}[3]{\ensuremath{W_{#3}(#2 | #1)}}
\newcommand{\empMI}[2]{\ensuremath{{I}\left(#1, #2\right)}}
\newcommand{\empcap}[1]{\ensuremath{\bar{C}(#1)}}

\newcommand{\hWempnl}[3]{\ensuremath{\hat{w}_r^{(#3)}(#2 | #1)}}
\newcommand{\hWempr}[1]{\ensuremath{\hat{W}_r^{(#1)}}}
\newcommand{\hWemprl}[3]{\ensuremath{\hat{W}_r^{(#3)}(#2 | #1)}}

\def\capmax{C_{\max}}

\def\rategap{\epsilon_{1}}

\def\rateloss{\rho}

\def\fbrate{\lambda}

\def\prob{\mathbb{P}}
\def\expe{\mathbb{E}}
\def\mbf{\mathbf}
\def\mc{\mathcal}

\def\argmax{\mathop{\rm argmax}}

\title{Zero-rate feedback can achieve the empirical capacity}
\author{\authorblockN{Krishnan Eswaran, Anand D. Sarwate, Anant Sahai, and Michael Gastpar} \\
\authorblockA{Department of Electrical Engineering and Computer Sciences\\
University of California, Berkeley\\
Berkeley, CA 94720, USA \\
Email: \{keswaran, asarwate, sahai, gastpar\}@eecs.berkeley.edu}%
\thanks{Manuscript received October XX, 2007;
revised XXXXXXXXXXXXXX.  }%
\thanks{Part of this work was presented at the 2007 International Symposium on Information Theory in Nice, France \cite{EswaranSSG:07isit}}
\thanks{K. Eswaran, A. Sahai, and M. Gastpar are with the Department of Electrical Engineering and Computer Sciences, University of California, Berkeley, Berkeley CA 94720-1770 USA. A. D. Sarwate was with the Department of Electrical Engineering and Computer Sciences, University of California, 
Berkeley. He is now with the Information Theory and Applications Center at the University of California, San Diego, La Jolla, CA 92093-0447 USA.}
\thanks{The work of A.D. Sarwate and M. Gastpar was supported in part by the National Science Foundation under award CCF-0347298.  The work of K. Eswaran, A. Sahai, and M. Gastpar was supported in part by the National Science Foundation under award CNS-0326503.}
}

\begin{document}

\maketitle

\begin{abstract}
The utility of limited feedback for coding over an individual sequence of DMCs is investigated. This study complements recent results showing how limited or noisy feedback can boost the reliability of communication. A strategy with fixed input distribution $P$ is given that asymptotically achieves rates arbitrarily close to the mutual information induced by $P$ and the state-averaged channel. When the capacity achieving input distribution is the same over all channel states, this achieves rates at least as large as the capacity of the state averaged channel, sometimes called the empirical capacity.
\end{abstract}

\section{Introduction}

Many contemporary communication systems can be modeled via a time-varying state.  For example, in wireless communications, the channel variation may be caused by neighboring systems, mobility, or other factors that are difficult to model.  In order to design robust communication strategies, engineers should adopt an appropriate model that can capture the channel dynamics.  One such model is the so-called arbitrarily varying channel (AVC), where the state can depend on the communication strategy and is selected in the worst possible manner.  One interpretation of this model is that there is a fixed rate that one wants to support over the worst possible channel states. An alternative and perhaps more relevant approach is an individual sequence model, where the state is fixed but unknown and not dependent on the communication strategy.  Here, a natural requirement is for a strategy to perform well whenever the state sequence is favorable, while for less favorable state sequences, inferior performance is acceptable. Essentially, this model considers the case in which one wants to adapt the rate to what the specific state sequence can support.

In order to achieve this variation in performance, the encoder must obtain some measure of the quality of the state sequence.  This requires additional resources, and the most natural model is to introduce \textit{feedback} from the receiver to the transmitter.  A second resource is joint \textit{randomization}
between the encoder and the decoder, which can also be enabled via feedback.  The encoder can use feedback to estimate the channel quality and hence communicate at rates commensurate with the channel quality.  Two fundamental questions are the following: first, how good a performance (in terms of achievable rate) can one expect for favorable state sequences? Second, how much feedback is required to attain this performance?  Many of the works in this area can be understood in terms of how they answer these two questions.

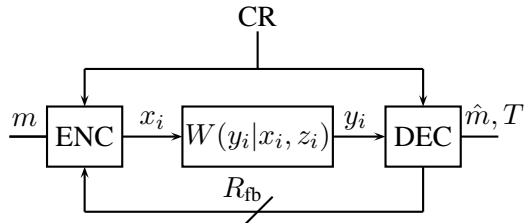
\begin{figure}[ht]
	\centering
\begin{pspicture}(-1.0,0)(5.5,2.3)

\rput(0.2,0.4){\pnode{SENSINK}}%
\rput(0.2,-0.2){\pnode{SENSOURCE}}%
\ncline{->}{SENSOURCE}{SENSINK}

\rput(1.5,0.8){\psframe(0,-0.4)(2.0,0.4)} %
\rput(2.5,0.8){$\textstyle{W(y_i|x_i,z_i)}$}%
\rput(1.5,0.8){\pnode{AtoCH}}%
\rput(3,0.8){\pnode{CHtoC}}%
\rput(1.1,1.02){$\textstyle{x_i}$}%
\rput(3.8,1.02){$\textstyle{y_i}$}%

\rput(-0.3,0.8){\psframe(0,-0.4)(1.0,0.4)} %
\rput(0.2,0.8){ENC}%
\rput(0.7,0.8){\pnode{A}}%

\rput(0.2,1.2){\pnode{ARA}}%

\rput(-0.8,0.8){\pnode{CRS}} %
\rput(-0.6,1.02){$\textstyle{m}$}%
\rput(-0.3,0.8){\pnode{CRT}}%
\ncline{CRS}{CRT}%

\rput(4.2,0.8){\psframe(0,-0.4)(1.0,0.4)} %
\rput(4.7,0.8){DEC}%
\rput(3.7,0.8){\pnode{C}}%

\rput(4.2,1.2){\pnode{ARC}}%

\rput(5.2,0.8){\pnode{CRR}} %
\rput(5.65,1.07){$\textstyle{\hat{m},T}$}%
\rput(5.7,0.8){\pnode{CRD}}%

\ncline{CRS}{CRT}%

\rput(3.5,0.8){\pnode{CHoutY}}%
\rput(4.2,0.8){\pnode{DECinY}}%
\rput(4.7,0.4){\psline(0,0)(0,-0.6)}%
\rput(4.7,-0.2){\psline(0,0)(-4.5,0)}%
\rput(2.5,-0.2){\psline(-.2,-.2)(.2,.2)}%
\rput(2.3,.1){$R_{\text{fb}}$}%

\rput(4.7,1.7){\psline(0,0)(-4.5,0)}%
\rput(4.7,1.7){\psline{->}(0,0)(0,-0.5)}%
\rput(0.2,1.7){\psline{->}(0,0)(0,-0.5)}%
\rput(2.5,2.2){\psline(0,0)(0,-0.5)}%
\rput(2.5,2.4){CR}%

\ncline{->}{CHoutY}{DECinY}

\ncline{CRR}{CRD}%

\ncline{->}{A}{AtoCH} 

\end{pspicture}
	\caption{Model setup with limited feedback and common randomness. \label{fig:blockdiagram}}
\end{figure}

The main trade-off for the channel model at hand is the correct balance between the resources
spent on communication versus those spent on channel estimation.
This trade-off is well understood in the case where the channel state sequence is fully
revealed to the receiver, as shown in the work of Draper {\it et al.} \cite{DraperFK:05rateless}.
Regarding the first question, for any fixed input distribution, their scheme can achieve rates arbitrarily close to the mutual information of the channel with the state known to both the transmitter and receiver.  They also provide an interesting answer to the second question: a feedback link of vanishing rate is sufficient to attain this performance.  To sum up, when channel estimation at the receiver is free, feedback of vanishing rate is enough.

Shayevitz and Feder \cite{ShayevitzF:06empirical} consider the more realistic case where the decoder has only the channel outputs. They develop a scheme in which the receiver keeps estimating the state sequence. In their consideration, the transmitter has full (causal) output feedback and can thus also track the state sequence. For the class of channels they consider, Shayevitz and Feder establish an achievable rate that they call the ``empirical capacity,'' which they define as the capacity of an i.i.d. channel with transition probabilities corresponding to the empirical statistics of the noise sequence.  Therefore if feedback is free, then rates arbitrarily close to  the ``empirical capacity'' are achievable.

This paper is a commentary on this development: we consider the same notion of ``empirical capacity,'' but provide an answer to the second question. Specifically, for a fixed input distribution, we show that 
if common randomness is available, a feedback link of vanishing rate is sufficient to achieve the empirical mutual information, which in some settings, such as the class of channels considered by Shayevitz and Feder, coincides with the ``empirical capacity''.  To do this, we adapt the feedback-reducing block/chunk strategies used earlier in the context of reliability functions \cite{Sahai:06delay,DraperAllerton06}, and most specifically in \cite{Sahai:07balancing}. They are in turn inspired by Hybrid ARQ \cite{Soljanin:03harq}.  Thus, the flavor of our algorithm is different from \cite{ShayevitzF:06empirical}.  By doing away with the output feedback, we lose the simplicity of the scheme in \cite{ShayevitzF:06empirical}, but we show that similar rates can still be obtained with almost negligible feedback.

The strategy developed in this paper fits in the category of rateless codes, which are a class of coding strategies that use limited feedback to adapt to unknown channel parameters.  Most studies about feedback for rate and reliability have centered around full output feedback \cite{Horstein:63feedback,SchalkwijkK:66feedback,Burnashev:76feedback,CoverP:89GaussianFeedback,OoiW:98iterative,Ooi:98feedback,Kim:06GaussianFeedback,Sahai:06delay}; however, recent work has started to improve our understanding of how limited feedback affects these performance measures.  
For instance, limited feedback can be used to improve reliability \cite{Sahai:07balancing}. 
Furthermore, in some multiuser Gaussian channels, noisy feedback increases the achievable rates 
\cite{GastparK:06noisy,Wigger:06kailath} and the reliability \cite{DraperAllerton06,KimLW:06noisyfeedback}. In a rateless code the decoder can use a low-rate feedback link to inform the encoder when it decodes.  These codes were first studied in the context of the erasure channel \cite{Luby:02LT,Shokrollahi:03fountain}.  Later work focused on compound channels \cite{Shulman:03commonbc,DraperFK:04isit,TchamkertenT:06variable}.  The work of Draper et al. \cite{DraperFK:05rateless} is to our knowledge the first step towards adapting rateless codes to time-varying states.

We are now in a position to compare the modeling assumptions in these previous works 
with the current investigation; the comparisons are summarized in Table \ref{tab:relwork}.  The initial studies of rateless coding by Shulman \cite{Shulman:03commonbc} and Tchamkerten and Telatar \cite{TchamkertenT:06variable} used feedback to tune the rate to the realized parameter governing the channel behavior.  The study of time-varying states was first introduced by Draper et al. \cite{DraperFK:05rateless}, but they assumed full state information at the decoder, which leads to higher rates.  Most recently, Shayevitz and Feder \cite{ShayevitzF:06empirical} showed an explicit coding algorithm based on Horstein's method \cite{Horstein:63feedback} that achieves the empirical capacity.  Their scheme uses full feedback, but in turn works for a larger class of
channel models. Moreover, it is a horizon-free scheme.

\begin{table}
\begin{center}
\begin{tabular}{l|c|c|c|c}%
  	&   channel model   &   feedback   &  state information & common randomness \\ %
\hline
Shulman \cite{Shulman:03commonbc}
			& compound      &   full    &  none & none \\ %
Tchamkerten and Telatar \cite{TchamkertenT:06variable}
			& compound      &   full    &  none & none \\ %
Draper, Frey, and Kschischang \cite{DraperFK:05rateless} 
			& AVC      &   0-rate    &  at decoder & none \\ %
Shayevitz and Feder \cite{ShayevitzF:06empirical}
			& individual sequence &   full    & none &yes \\ %
This paper 
			&  individual sequence &   0-rate    & none & yes %
\end{tabular}
\end{center}
\caption{Related results and assumptions on channel model, feedback, state information and common randomness \label{tab:relwork}}
\end{table}

In our scheme, the encoder attempts to send $k$ bits over the channel during a variable-length \textit{round}.  The encoder sends \textit{chunks} of the codeword to the decoder, after which the decoder feeds back a decision as to whether it can decode.  The encoder and decoder use common randomness to choose a set of randomly chosen \textit{training} positions during which the encoder sends a pilot sequence.  The decoder uses the training positions to estimate the channel.  As soon as the total \textit{empirical mutual information} over the aggregate channel sufficiently exceeds $k$ bits, the decoder attempts to decode.  Through this combination of training-based channel estimation and robust decoding we can exploit the limited feedback to achieve rates asymptotically equal to those with advance knowledge of the average channel.  

In the next section, we motivate the study of this problem with some concrete examples. In Section \ref{sec:algorithm}, we define the channel model, state our main result, and describe the coding strategy.  Section \ref{sec:analysis} contains the analysis of our strategy with most of the technical details reserved for the Appendix.

\section{Motivating Examples}

The following two simple examples will prove useful in explaining the meaning of the main result of this paper, and help motivate the present study.  The first is the model considered in \cite{ShayevitzF:06empirical} -- a binary modulo-additive channel with a noise sequence whose empirical frequency of $1$'s is unknown.  In this example, the ``empirical mutual information'' under all state sequences is maximized by the uniform distribution, so our algorithm achieves the ``empirical capacity''.  In the second example we consider the $Z$-channel for which the input distribution maximizing the empirical mutual information is not identical for all state sequences, so our scheme will not in general achieve rates as high as the empirical capacity.

\subsection{Binary modulo-additive channels \label{ex:binary}}

The simplest example of a channel with an individual noise sequence is the binary modulo-additive channel. This channel takes binary inputs and produces binary outputs, where the output is produced by flipping some bits of the channel input. These flips do not depend on the channel input symbols. The output $\mathbf{y} \in \{0,1\}^N$ can be written as 
    \begin{align}
    \mathbf{y} = \mathbf{x} \oplus \mathbf{z},
    \end{align}
where $\mathbf{x} \in \{0,1\}^N$ is the channel input, $\mbf{z} \in \{0,1\}^N$ is the noise sequence, and addition is carried out modulo-$2$.  The noise $\mbf{z}$ is arbitrary but fixed, and we let $p$  be the empirical fraction of $1$'s in $\mbf{z}$, which is arbitrary but fixed over the $[0,1]$ interval. 

Because the state sequence $\mbf{z}$ is arbitrary and unknown, it is not clear how to find the highest possible rate of reliable communication.  For any \textit{fixed} $\mbf{z}$, we could say na\"{i}vely that the capacity is one bit, because the channel is deterministic.  However, $\mbf{z}$ is unknown and may, in fact, have been generated iid according to a Bernoulli distribution with parameter $p$, in which case the capacity should be no larger than $1 - \hb(p)$, namely, the capacity of a binary symmetric channel (BSC) with crossover $p$.  The algorithm in this paper guarantees a rate close to $1 - \hb(p)$ for any state sequence $\mbf{z}$ with an empirical fraction of $1$'s equal to $p$.  This rate can be thought of as the empirical mutual information of the channel with input distribution $(1/2, 1/2)$.  Since the input distribution is the same for all BSC's, the rate can also be called the \textit{empirical capacity}, as in the work of Shayevitz and Feder \cite{ShayevitzF:06empirical}.

\subsection{Z-channels with unknown crossover \label{ex:zchan}}

Whereas the example above can be thought of as an XOR operation with the channel state, in our second example, we consider a binary channel in which the output is the logical OR of the input and state. For input $x$ and noise $y$, the output is given by the following:
	\begin{align}
	y = \left\{
		\begin{array}{ll}
		x  &  z = 0 \\
		1  &  z = 1~.
		\end{array}
	\right.
	\end{align}
Again, the noise sequence $\mbf{z}$ is arbitrary but fixed.  Let $q$ denote the empirical fraction of $1$'s in $\mbf{z}$. 

The algorithm in this paper achieves rates close to those corresponding to a Z-channel with crossover probability $q$.  The channel is the average $W_{\mbf{z}}$ of $W(y | x, z_i)$ over $\mbf{z}$.   Unlike the previous examples, this channel has a capacity achieving input distribution that depends on $q$.  The algorithm proposed in this paper chooses a fixed input distribution $P$ and achieves the mutual information $I(P,W_{\mbf{z}})$ of a Z-channel with that input distribution.  This leaves open the question of how to choose $P$.  One method is to choose the $P$ that minimizes the gap between $\max_{Q} I(Q,W_{\mbf{z}}) - I(P,W_{\mbf{z}})$ over all $\mbf{z}$.  However, in many cases the uniform distribution is not a bad choice, as shown by Shulman and Feder \cite{ShulmanF:04prior}.  In our results we leave the choice of $P$ open for the designer.

\section{The channel model and coding strategy \label{sec:algorithm}}

\subsection{Notation}

Script letters will generally be used to denote sets and alphabets and boldface to denote vectors.  For a vector $\mbf{x} = (x_1, x_2, \ldots, x_n)$, we write $\mbf{x}_{i}^{j}$ for the tuple $(x_i, x_{i+1}, \ldots, x_j)$ and $\mbf{x}^j$ for the tuple $(x_1, x_2, \ldots, x_j)$.  The notation $[J]$ will be used as shorthand for the set $\{1, 2, \ldots, J\}$.  The probability distribution $\styp{\mbf{z}}$ is the type of a sequence $\mbf{z}$.  For a distribution $Q$, the set $\typ{N}{Q}$ is the set of all length $N$ sequences of type $Q$.

\subsection{Channel model and coding}

The problem we consider in this paper is that of communicating over a channel with an individual state sequence.  Let the finite sets $\mc{X}$ and $\mc{Y}$ denote the channel input and output alphabets, respectively.  The channel model we consider consists of a family of channels $\mc{W} = \{W(y | x, z) : z \in \mc{Z}\}$ indexed by a state variable in a finite set $\mc{Z}$.  For any state sequence $\mbf{z} = (z_1, z_2, \ldots, z_N)$, and output $y_i$, we assume
    \begin{align*}
    \prob(y_i | \mbf{x}^i, y^{i-1}, \mbf{z}) = W(y_i | x_i, z_i)~.
    \end{align*}
That is, the channel output depends only on the current input and state.

We consider coding for this channel using the setup shown in Figure \ref{fig:blockdiagram}.  We think of the rate-limited feedback link as a noiseless channel that can be used every $\nfb$ uses of the forward channel to send $\Bfb$ bits.  The rate of the feedback is thus $\Rfb = \Bfb/\nfb$.  To avoid integer effects, we will consider only integer values for $\nfb$ and $\Bfb$.  We assume that the encoder and decoder have access to a common random variable $G$ distributed uniformly over the unit interval $[0,1]$.  This random variable can be used to generate common randomness that is shared between the encoder and decoder.

Because the maximum capacity of this set of channels is $\capmax = \log \min\{|\mc{X}|,|\mc{Y}|\}$, we define the set of possible messages to be the set of all binary sequences $\{0,1\}^{N \capmax}$.  This message set is naturally nested -- the truncated set $\{0,1\}^{T}$ is a set of prefixes for $\{0,1\}^{N \capmax}$.  At the time of decoding, the decoder will decide on a decoding threshold $T \in \mathbb{N}$ and a message $m \in \{0,1\}^{T}$.  The threshold $T$ is itself a random variable that will depend on the state sequence $\mbf{z}$, the common randomness $G$, and the randomness in the channel.

An $(N, \nfb, \Bfb)$ \textit{coding strategy} for blocklength $N$ consists of a sequence of (possibly random) encoding functions for $i = 1, 2, \ldots, N$,
    \begin{align}
    \eta_{i} : \{0,1\}^{N \capmax} \times \{0,1\}^{\lfloor (i-1)/\nfb \rfloor \Bfb} \times [0,1] \to \mc{X}~,
    \end{align}
a sequence of (possibly random) feedback functions for $i = \nfb, 2 \nfb, \ldots$:
    \begin{align}
    \phi_{i} : \mc{Y}^{i} \times [0,1] \to \{0,1\}^{\Bfb}~,  %
    \end{align}
and a decoding function
    \begin{align}
    \psi : \mc{Y}^N \times [0,1] \to \{0, 1, \ldots, N \capmax\} \times \{0,1\}^{N \capmax}~.
    \end{align}

We say a message $\mbf{m} \in \{0,1\}^{N \capmax}$ is \emph{encoded} into a codeword $\mbf{x} \in \mc{X}^N$ if
    \begin{align}
    x_i = \eta_{i}(\mbf{m},\phi_1(y^{\nfb}, G), \ldots, \phi_{\lfloor (i-1)/\nfb \rfloor}(y^{\lfloor (i-1)/\nfb \rfloor \cdot \nfb }, G),G) \qquad \forall i \in [N]~.
    \end{align}
For an $(N, \nfb, \Bfb)$ coding strategy, let $\psi(\mbf{y},G) = (T,\mbf{\hat{m}})$.  The first output $T \in \{0, 1, \ldots, N \capmax\}$ is the \emph{decoding threshold} and $\mbf{\hat{m}}^T$ is the \emph{message estimate}.  Both of these quantities are random variables. 

For a state sequence $\mbf{z}$, the \textit{maximal error probability} of an $(N, \nfb, \Bfb)$ coding strategy, is defined as
    \begin{align}
    \schemeerr({\mbf{z}}) = \max_{\mbf{m} \in \{0,1\}^{N \capmax}}
        \prob_{G,\mc{W}} \left(
                \mbf{m}^T \ne \mbf{\hat{m}}^T
            \ \Big|\  \mbf{z}
            \right)~.
    \end{align}
where the probability is taken over the common randomness $G$ and randomness in the channel.  For a state sequence $\mbf{z}$, a \emph{rate} $R$ is said to be \emph{achievable} with probability $1 - \rateerr({\mbf{z}})$ if
        \begin{align}
        \rateerr({\mbf{z}}) &=
            \max_{\mbf{m} \in \{0,1\}^{N \capmax}}
            \prob_{G,\mc{W}} \left(
            R \geq T/N,~\mbf{m}^T \ne \mbf{\hat{m}}^T
              \ \Big|\  \mbf{z}, \mbf{m}
            \right)~.
        \end{align}
Note that we can upper bound $\rateerr({\mbf{z}})$ :
        \begin{align}
        \rateerr({\mbf{z}}) &\le \schemeerr({\mbf{z}})
                + \max_{\mbf{m} \in \{0,1\}^{N \capmax}}
                    \prob_{G,\mc{W}} \left(
                R \geq T/N
                \ \Big|\  \mbf{z}, \mbf{m}
            \right)~.
        \end{align}

Note that this channel model assumes a known finite horizon $N$, unlike the infinite horizon model of Shayevitz and Feder \cite{ShayevitzF:06empirical}.  Furthermore, the basic model assumes an unbounded amount of common randomness in the form of the real number $G$.
This point is discussed further in Section~\ref{sec:discussion}.

\subsection{Mutual information definitions \label{Sec-MIEC}}

The results in this paper are stated in terms of mutual information quantities involving time-averaged channels dependent on the individual state sequence $\mbf{z}$.  For fixed $\mbf{z}$ define the \textit{state-averaged channel} to be
    \begin{align}
        \empchan{x}{y}{\mbf{z}} = \frac{1}{N} \sum_{i=1}^{N} W(y | x, z_i)~. \label{eq:avgchan}
    \end{align}
Note that if $\mbf{z}$ and $\mbf{z}'$ have the same type, then the state-averaged channels generated by them are the same.  Define the empirical channel for a distribution $Q$ on $\mc{Z}$:
    \begin{align}
        \empchan{x}{y}{Q} = \sum_{z \in \mc{Z}} W(y | x, z) Q(z)~.
    \end{align}

For a fixed input distribution $P(x)$ on $\mc{X}$ and channel $W(y | x)$, the \textit{mutual information} is given by the usual definition:
    \begin{align*}
    \empMI{P}{W} = \sum_{x,y} W(y | x) P(x) \log \frac{W(y | x) P(x)}{ P(x) \sum_{x'} W(y | x') P(x')}~.
    \end{align*}
For an individual state sequence $\mbf{z}$ the \textit{empirical mutual information} is given by $\empMI{P}{W_{\mbf{z}}}$.

\subsection{Optimality versus empirical capacity}

We are interested in analyzing strategies that can adapt their rates depending on the state sequence, and in our analysis, we want to consider the rates achieved by a strategy as a function of the state sequence. Unlike the compound channel setting (see e.g. \cite{CsiszarKorner:82book} for definitions), which considers the worst-case behavior of a strategy over a class of channels, we instead want strategies that perform universally well over all sequences. However, this raises the problem of finding a notion of optimality that does not depend on the worst-case performance.

One possibility is to define an optimal strategy as one that, for every state sequence, achieves a rate at least as large as any other strategy for that sequence, and then define the capacity as the rates achieved by this strategy.  However, this means comparing a strategy for all sequences against all strategies tailored to a fixed sequence.  In the example in Section \ref{ex:binary}, for each $\mbf{z}$ there exists a decoding strategy which adds $\mbf{z}$ to the output, undoing all of the bit flips.  Each strategy achieves rate $1$ for the specific choice of $\mbf{z}$, but this is clearly an unreasonable target.

Instead, for each sequence we can consider a set of reference strategies and measure the ``regret'' of our strategy with respect to the reference strategies for each sequence.  We take an approach inspired by source coding for individual sequences, in which we have a benchmark rate for each state sequence and then test whether a coding strategy attains the benchmark for each state sequence.

One such benchmark that we consider in this paper is the \textit{empirical capacity} -- for a fixed $\mbf{z}$, the empirical capacity is defined as the supremum over all input distributions of the empirical mutual information:
    \begin{align*}
    \empcap{\mbf{z}} = \sup_{P(x)} \empMI{P}{W_\mbf{z}}~.
    \end{align*}
First used by Shayevitz and Feder \cite{ShayevitzF:06empirical}, 
empirical capacity is given its name not because it is purported to be optimal, but instead because of its resemblance to the capacity of a point-to-point discrete memoryless channel. 

There are two points that are worth mentioning before proceeding to describe the results in this paper.  First, it is easy to see that the empirical capacity is a weaker target than the best possible strategy for a given sequence.   It is possible that a strategy can achieve rates larger than the empirical capacity. In the example in Section \ref{ex:binary}, if the sequence $\mbf{z}$ were all $0$ for the first half and all $1$ for the second half, the empirical capacity is $0$, whereas the coding strategy presented in this paper is expected to achieve rates close to $1$.

Second, there may exist examples for which no strategy is guaranteed to achieve the empirical capacity. The coding strategy proposed in this paper uses a fixed input distribution $P$, and in general, the maximizing $P(x)$ may not be the same for all $\mbf{z}$.\footnote{A question then arises of how one chooses 
the input distribution $P$. One possibility could be 
to choose $P$ to be uniform over the input alphabet. However, depending on the setting, 
other approaches might be preferable. Inspired by the theory of AVCs, one may choose the input
distribution to be
    \begin{align}
    P = \argmax_{P'} \inf_{ Q : \empMI{P'}{W_Q} > \rateloss } \empMI{P'}{ W_{Q} }~,\label{Eq-maxminPX}
    \end{align}
where $\rateloss$ is a parameter governing the gap between the rates guaranteed by the algorithm and the empirical mutual information of the channel.  This approach can run into problems in some situations 
in which for the $P$ chosen, $\empMI{P}{ W_{Q} } = 0$ for a large subset of state distributions $Q$, but there exists a distribution $\tilde{P}$ for which $\empMI{\tilde{P}}{ W_{Q} } \geq \rateloss$ for all $Q$. On the other hand, if one were to remove the condition that $\empMI{P'}{W_Q} > \rateloss$, for the example in Section \ref{ex:binary}, $\inf_{ Q } \empMI{P'}{ W_{Q} } = 0$ for 
all choices of $P'$, and the choice of $P'$ would be arbitrary. Because of such issues, we will leave the question of how to choose the input distribution $P$ unanswered in this work. The problem of choosing $P$ is similar to that studied by Shulman and Feder \cite{ShulmanF:04prior}.}
In these cases our strategy can achieve rates close to the empirical mutual information $\empMI{P}{W_\mbf{z}}$ but not 
the empirical capacity $\empcap{\mbf{z}}$.  It may be possible to adapt $P$ over time, and finding a strategy achieving $\empcap{\mbf{z}}$ or a counterexample showing that for some channels, no strategy achieving $\empcap{\mbf{z}}$ is possible, is left for future research.

\subsection{Main result}

The main result in this paper is that the algorithm given in the next section achieves rates that asymptotically approach the mutual information $\empMI{P}{\empchanshort{\mbf{z}}}$ for a large set of state sequences $\mbf{z}$.

\begin{theorem} \label{theorem:main_result}
Let $\{W(y|x,z) : z \in \mc{Z}\}$ be a given family of channels.  Then given any $\rateloss > 0$, $\algerror > 0$, $\fbrate^\ast > 0$, and channel input distribution $P$, there exists an $N$ sufficiently large and an $(N, \nfb, \Bfb)$ coding strategy with feedback rate 
	\begin{align}
		\Rfb = \frac{\Bfb}{\nfb} < \fbrate^\ast~,
	\end{align}
such that for all $\mbf{z} \in \typ{Q}{N}$, the rate
    \begin{align}
    R \ge \empMI{P}{\empchanshort{Q}} - \rateloss~
    \label{eq:thm_rateloss}
    \end{align}
is achievable with probability $1 - \algerror$.
\end{theorem}

\textit{Binary modulo-additive channels, revisited}:  For the binary additive example in Section \ref{ex:binary}, $p$ denoted the fraction of ones in the noise sequence $\mbf{z}$. Then, the empirical capacity is $1 - \hb (p)$, the capacity of the binary symmetric channel with crossover probability $p$.  Theorem \ref{theorem:main_result} implies the existence of strategies employing asymptotically zero-rate feedback such that for all $\rateloss, \algerror > 0$ and sufficiently large $N$,
    \begin{align}
    R \ge 1 - \hb (p) - \rateloss ~,
    \end{align}
is achievable with probability at least $1 - \algerror$.

\textit{Z-channels with unknown crossover, revisited}: For the example in Section \ref{ex:zchan} with $q$ equal to the fraction of $1$'s in the crossover sequence, the capacity achieving input distribution is a function of $q$, so the theorem cannot guarantee a scheme achieving the empirical capacity.  Despite this, it still provides achievable rates in this setting.  If the channel input distribution has $P(X = 1) = p_x$ for this channel, then the empirical mutual information for this channel can be written as
    \begin{align}
    \empMI{P}{W_q} = \hb(p_x) - (1 - p_x + p_xq)
            \hb\left( \frac{p_xq}{1-p_x +p_xq} \right)~,
    \end{align}
and is asymptotically achievable from Theorem \ref{theorem:main_result}.
As discussed briefly at the end of Section~\ref{Sec-MIEC},
the question of how to select $p_x$ is outside the framework of
this paper; one possibility is given in equation \eqref{Eq-maxminPX}.

\subsection{Proposed coding strategy: Randomized rateless code \label{sect:proposed_strategy}}

The achievability result in Theorem \ref{theorem:main_result} relies on the following coding strategy, which can be thought of as iterated rateless coding with randomized training (or, for short, randomized rateless code).  The overall scheme is illustrated in Figure \ref{fig:scheme}.  The scheme divides time into chunks of $\chunk(N)$ channel uses and in each round attempts to send $\rbits(N)$ using a randomized rateless code.  Each chunk contains randomly places training sequences so the decoder can estimate the empirical channel.  The decoder chooses to decode when the empirical rate falls below the estimated empirical mutual information calculated from the channel estimates.  After the $\rbits(N)$ are decoded the round ends and the encoder starts a new round to send the next $\rbits(N)$ bits.  The length of each round is variable and depends on the empirical state sequence.

We now describe each component of the scheme in more detail.

\subsubsection{Feedback}

Divide the blocklength $N$ into \textit{chunks} of length $\chunk = \chunk(N)$.  Feedback occurs at the end of chunks, so $\nfb = \chunk$ with three possible messages: ``BAD NOISE,'' ``DECODED,'' and ``KEEP GOING,'' which correspond to the feedback messages $00$, $01$, and $10$, respectively.  Thus, $\Bfb = 2$, so the feedback rate $\Rfb = \fbrate (N)$ is given by the expression
    \begin{align} \label{eq:strategy_fbrate}
        \Rfb = \frac{\Bfb}{\chunk(N)}~.
    \end{align}
If the chunk size $\chunk (N)$ goes to infinity as $N \to \infty$, the feedback rate $\fbrate (N) \to 0$.

\subsubsection{Rateless coding}

A rateless code is a variable-length coding scheme to send a fixed number of bits.  In the algorithm proposed here, the encoder attempts to send $\rbits = \rbits(N)$ bits over several chunks comprising a \textit{round}.  Rounds vary in length and terminate at the end of chunks in which the decoder feeds back either ``BAD NOISE'' or ``DECODED.''  Let $\ell_r$ denote the time index at the end of round $r$:
    \begin{align} \label{eq:roundlength_defn}
        \ell_r = \min \left\{ j = i \cdot b(N) > \ell_{r - 1} : \phi_i  = \text{``BAD NOISE'' or ``DECODED''} \right\}~,
    \end{align}
and set $\ell_0 = 0$.

An $(\mhi, \codechunk, \rbits)$ \textit{rateless code} is a sequence of maps $\{ (\mu_i,\nu_i) : i = 1, 2, \ldots \mhi \}$, where
    \begin{align} \label{eq:rateless_defn}
    \mu_i : \{0,1\}^{\rbits} &\to \mc{X}^{\codechunk} \\
    \nu_i : \mc{Y}^{i \cdot \codechunk} &\to \{0,1\}^{\rbits}~.
    \end{align}
The encoding maps $\mu_i$ produce successive chunks of a codeword for a given message, and the decoding maps attempt to decode the message based on the channel outputs.  An $(\mhi, \codechunk, \rbits)$ \textit{randomized} rateless code is a random variable that takes values in the set of $(\mhi, \codechunk, \rbits)$ rateless codes.  The \textit{maximal error probability} $\ratelesserr(M,\mbf{z}) = \ratelesserr(M,\mbf{z}, \mathcal{D})$ for a randomized rateless code $\mathcal{D}$ decoded at time $M \codechunk$ with state sequence $\mbf{z} \in \mc{Z}^{M \codechunk}$ is
    \begin{align} \label{eq:randomized_rateless_defn}
    \ratelesserr(M,\mbf{z}, \mathcal{D}) &= \max_{m \in \{0,1\}^{\rbits}}
        \expe\left[ W^{M \codechunk} \left(
                \left\{
                    \nu_M( \mbf{y}_{1}^{M \codechunk} ) \ne m
                    \right\}
                ~\Big|~
                \mu_i(m), \mbf{z}
            \right)
            \right] \\
            &= \max_{m \in \{0,1\}^{\rbits}} \varepsilon_m (M, \mbf{z}, \mc{D})~, \label{eq:randomized_rateless_msgerr_defn}
    \end{align}
where the expectation is taken over the randomness in the code.  We will suppress
dependence on $\mathcal{D}$ when it is clear from context.  The randomized rateless code used in this paper has codewords with constant composition $P(x)$ on $\mc{X}$ and uses a maximum mutual information (MMI) decoder.

\subsubsection{Training \label{sect:training}}

The coding strategy analyzed in this paper uses a randomized
rateless code in conjunction with randomly located training symbols.
The training allows the decoder to estimate the channel and choose
an appropriate decoding time.  For each chunk of $\chunk$ channel
uses, the scheme uses $\train = \train(N)$ positions for training.
Using the common randomness $G$, the encoder and decoder select
$\train$ \textit{training positions} $\itrain_{r,n}$ for the $n$-th
chunk of round $r$.\footnote{There is a slight abuse of notation 
with the type $\typ{N}{Q}$, but the double subscript in $\itrain_{r,n}$ should make the distinction  unambiguous.}  Formally, $\itrain_{r,n}$ is uniformly distributed
over subsets of $\{\ell_{r-1} + (n-1) \chunk+ 1,\ldots, \ell_{r-1} +
n \chunk\}$ of cardinality $\train$. This set is further randomly
partitioned into $|\mc{X}|$ subsets $\itrain_{r,n}(x)$ for $x \in
\mc{X}$.

\subsubsection{Encoding}

The encoder attempts to send a message $m \in \{0,1\}^{N \capmax}$ over several rounds.  In each round it attempts to send a sub-message $m_r \in \{0,1\}^{\rbits}$ consisting of $\rbits$ bits of $m$.  The submessage $m_1$ is the first $\rbits$ bits of $m$.  If the round $r-1$ ended with ``BAD NOISE'' then $m_r = m_{r-1}$, and if round $r-1$ ended with ``DECODED'' then $m_r$ is the next $\rbits$ of the message $m$.

The encoder and decoder share an $(\mhi, \chunk - t, \rbits)$ randomized rateless code.  Using the common randomness $G$, at the start of each round the encoder and decoder choose an $(\mhi, \chunk - t, \rbits)$ rateless code $\{ \mu_j, \nu_j) : j = 1, 2, \ldots \mhi \}$ according to the distribution of this randomized code.  Define the encoding map $\eta_i$ in the $n$-th chunk of the $r$-th round:
    \begin{align}
    \eta_i(m_r,G) &= x  \qquad i \in \itrain_{r,n}(x)  \label{eq:encdef1} \\
    \{ \eta_i(m_r,G) : i \notin \itrain_{r,n} \} &= \mu_n(m_r)~. \label{eq:encdef2}
    \end{align}
That is, the $n$-th chunk transmitted by the scheme is created by taking the $\chunk - t$ piece of the codeword $\mu_n(m_r)$ and inserting the $t$ randomly chosen training positions, as illustrated in Figure \ref{fig:scheme}.  The dependence of $\eta_i$ on the feedback is suppressed here because a round $r$ is terminated as soon as the feedback message is no longer ``KEEP GOING.'' 

    \begin{figure*}[thb]
    \begin{center}
\psset{unit=0.7mm}
\begin{pspicture}(-16,-19)(134,40)

\newrgbcolor{dkgray}{0.2 0.2 0.2}
\newrgbcolor{ltgray}{0.8 0.8 0.8}
\newrgbcolor{medgray}{0.5 0.5 0.5}

\rput(-8,32.5){total:}
\psframe[fillstyle=solid,fillcolor=dkgray](0,30)(30,35)
\psframe[fillstyle=solid,fillcolor=dkgray](32,30)(50,35)
\psframe[fillstyle=solid,fillcolor=dkgray](52,30)(75,35)
\psline(0,30)(85,30)
\psline(0,35)(85,35)
\psdots(77,32.5)(80,32.5)(83,32.5)
\psframe[fillstyle=solid,fillcolor=dkgray](100,30)(130,35)
\psline(90,30)(130,30)
\psline(90,35)(130,35)
\psdots(98,32.5)(95,32.5)(92,32.5)
\rput(0,38){$0$}
\rput(130,38){$N$}

\rput(15,27){\textcolor{dkgray}{\textbf{ \small round 1}}}
\rput(41,27){\textcolor{dkgray}{\textbf{ \small 2}}}
\rput(63.5,27){\textcolor{dkgray}{\textbf{ \small 3}}}

\psline[linewidth=1pt]{->}(0,30)(0,18)
\psline[linewidth=1pt]{->}(30,30)(30,23)(110,23)(110,18)

\rput(-8,12.5){round:}
\psframe[fillstyle=solid,fillcolor=medgray](0,10)(20,15)
\psframe[fillstyle=solid,fillcolor=medgray](25,10)(45,15)
\psframe[fillstyle=solid,fillcolor=medgray](50,10)(70,15)
\psdots(75,12.5)(80,12.5)(85,12.5)
\psframe[fillstyle=solid,fillcolor=medgray](90,10)(110,15)
\rput(20,18){$\chunk$} \rput(45,18){$2 \chunk$} \rput(70,18){$3 \chunk$}

\rput(10,12){\textcolor{medgray}{\textbf{\small chunk 1}}}
\rput(35,12){\textcolor{medgray}{\textbf{\small 2}}}
\rput(60,12){\textcolor{medgray}{\textbf{\small 3}}}

\psline[linewidth=1pt]{->}(0,10)(0,0)
\psline[linewidth=1pt]{->}(20,10)(20,8)(75,8)(75,0)

\rput(-8,-7.5){chunk:}
\psframe[fillstyle=solid,fillcolor=ltgray](0,-5)(75,-10)
\psframe[fillstyle=solid,fillcolor=black](10,-5)(12,-10)
\psframe[fillstyle=solid,fillcolor=black](50,-5)(52,-10)
\psframe[fillstyle=solid,fillcolor=black](4,-5)(6,-10)
\psframe[fillstyle=solid,fillcolor=black](63,-5)(65,-10)
\rput(75,-2){$\chunk$}

\rput[B](10,-18){\textbf{training}}
\psline{->}(9,-14)(5,-11)
\psline{->}(11,-14)(11,-11)

\rput[B](45,-18){\textbf{codeword}}
\psline{->}(44,-14)(31,-11)
\psline{->}(46,-14)(57.5,-11)

\end{pspicture}    \caption{After each chunk of length $b$ feedback can be sent.  Rounds end by decoding a message or declaring the noise to be bad. \label{fig:scheme}}
    \end{center}
    \end{figure*}
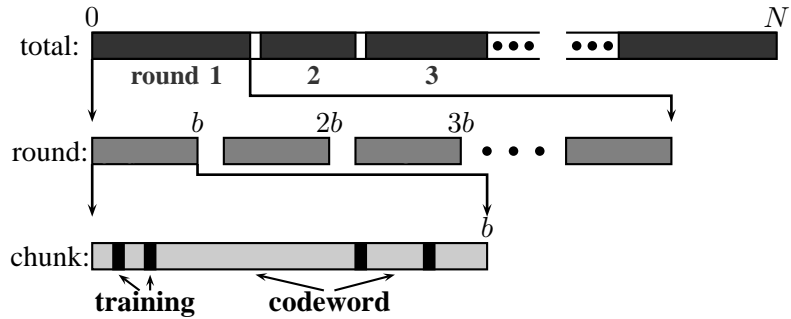

\subsubsection{Decoding}

The decoder uses the training symbols $\{y_i : i \in \itrain_{r,n}(x)\}$ to estimate the channel transition probabilities $W_{\mbf{z}}(y | x)$ and thereby obtain an estimate of the empirical mutual information $\empMI{P}{W_{\mbf{z}}}$ during the chunk and over the round.  If the estimated mutual information is too low, then it feeds back ``BAD NOISE.''  If the estimated mutual information is above the empirical rate $\rbits/(n (\chunk - \train)) + \rategap$ then it decodes the code using the MMI decoder $\nu_n$ of the rateless code and feeds back ``DECODED.'' Otherwise, it feeds back ``KEEP GOING.''  The parameter $\rategap$ ensures that with high probability the empirical rate is below the true empirical mutual information of the channel.

\subsubsection{Algorithm \label{subsec:algorithm}}

The parameters of the algorithm are a the chunk size $\chunk(N)$, training size $t(N)$, number of bits per round $\rbits$, and decoding thresholds $\rategap$ and $\badrate$.

Given an $(\mhi, \chunk - \train, \rbits)$ randomized rateless code and message bits $m_r$, the encoder and decoder first use common randomness to choose a realization of the randomized rateless code. The following steps are then repeated for each chunk in round $r$:
\begin{enumerate}
    \item Using common randomness, the encoder and decoder choose $\train = \train (N)$ positions $\itrain_{r,n}$ and a random partition of $\itrain_{r,n}$ into $|\mc{X}|$ subsets $\itrain_{r,n}(x)$ of size $t/|\mc{X}|$ for training in chunk $n$.
\item \label{alg2:step1} The encoder transmits the $n$-th chunk using the encoding map as defined in Equations (\ref{eq:encdef1})--(\ref{eq:encdef2}).  In particular, the symbol $x$ is sent during the training positions $\itrain_{r,n}(x)$.
    \item The decoder estimates the empirical channel in chunk $n$ and the empirical channel over the round so far:
        \begin{align}
            \hWempnl{x}{y}{n} &= \frac{|\mc{X}|}{\train} \cdot
                \left| \left\{ j \in \itrain_{r,n}(x) : y_j = y \right\} \right|
                \label{eq:chunk_estimate} \\
            \hWemprl{x}{y}{n} &= \frac{1}{n} \sum_{i=1}^n \hWempnl{x}{y}{i}~.
                \label{eq:round_estimate}
        \end{align}
\item The decoder makes a decision based on $\hWempr{n}$ and $n$.
    \begin{enumerate}
    \item If
        \begin{align} \label{eq:badnoise_condition}
        \empMI{P}{\hWempr{n}} - \rategap < \badrate~,
        \end{align}
        where $\badrate > 0$ is a parameter of the algorithm, then the decoder feeds back ``BAD NOISE'' and the round is terminated without decoding the $\rbits$ bits.  In the next round, the encoder will attempt to resend the $\rbits$ bits from this round.
    \item If
        \begin{align}
         \empMI{P}{\hWempr{n}} - \rategap & > \frac{\rbits}{(\chunk - \train) \times n}~,
            \label{eq:decrule}
        \end{align}
        where $\train$ is defined in Section \ref{sect:training}, then the decoder decodes, feeds back ''DECODED,'' and the encoder starts a new round.
    \item otherwise the decoder feeds back ``KEEP GOING'' and goes to
    \ref{alg2:step1}).
    \end{enumerate}
    Thus, we have that
        \begin{align}
        \phi(\mbf{y},G) = \left\{ 
			\begin{array}{cl} 
			\text{``BAD NOISE'' }, & 
				\empMI{P}{\hWempr{n}} - \rategap < \badrate,\ 
						\empMI{P}{\hWempr{n}} - \rategap 
							\leq \frac{\rbits}{(\chunk - \train) \times n} \\
        		\text{``DECODED'' }, &
             	\empMI{P}{\hWempr{n}} - \rategap 
					> \frac{\rbits}{(\chunk - \train) \times n} \\
        		\text{``KEEP GOING'' }, & 
				\text{otherwise}
        		\end{array} 
			\right.
        \label{eq:decrules}
    		\end{align}
\end{enumerate}

This strategy has two main ingredients.  First, the encoder uses random training sequences to let the decoder accurately estimate the empirical average channel.  Given this accurate estimate, the decoder can track the empirical mutual information of the channel over the round.  Second, the decoder only needs to know that the empirical rate is smaller than the empirical mutual information in order guarantee a small error probability.  

We note again that the channel model and problem formulation involve a fixed overall blocklength $N$ and other parameters of the coding strategy are defined in terms of this parameter.  However, in practice it may be more desirable to fix a number of bits $\rbits(N)$ to send per round and then define the coding parameters in terms of $\rbits$.  We have chosen the former method because it is convenient for our mathematical analysis, but we believe that in principle the problem could be formulated in an ``infinite-horizon'' manner as well.
\section{Analysis \label{sec:analysis}}

Showing that the strategy proposed in the previous section satisfies the conditions of Theorem \ref{theorem:main_result} requires some more notation.  For each round $r$, let the random variable $M(r)$ be the number of chunks in that round:
    \begin{align}
    M(r) = \inf_{n > 0} \left\{ \empMI{P}{\hWempr{n}} - \rategap < \badrate
        \ \ \text{or}\ \
        \frac{\rbits}{(\chunk - \train)n}
            < \empMI{P}{\hWempr{n}} - \rategap \right\}~.
    \label{eq:M_defn}
    \end{align}
Let $\icode_{n,r}$ denote the time indices in the $n$-th chunk of round $r$ that are not in the training set $\itrain_{n,r}$.

The scheme depends on a number of parameters -- the overall
blocklength $N$, the number of bits per round $\rbits(N)$, the chunk
size $\chunk(N)$, the number of training positions per chunk
$\train(N)$, the \textit{rate gap} $\rategap(N)$, the error bound
$\algerror$, and the feedback rate $\fbrate(N)$.  In order to make
the proof of the result clear, assume that there exist real
constants $g_1, g_2, g_3 \in (0,\frac{1}{2})$ with $g_1 > g_2 > g_3$ and set
    \begin{eqnarray} \label{eq:param_assump}
        \rbits(N) = \Theta(N^{2 g_1}), & \chunk(N) = \Theta(N^{g_2}), & \train(N) = \Theta(N^{g_3})~.
    \end{eqnarray}
In particular, this means that the ratios $\rbits(N)/N \to 0$,
$(\chunk(N))^2/\rbits(N) \to 0$, and $\train(N)/\chunk(N) \to 0$.

\subsection{Error events \label{sect:error_events}}

The scheme requires that the channel estimates $\hWempr{M(r)}$ given
in (\ref{eq:round_estimate}) be ``close'' to the channel averaged
over the non-training positions $\icode_{n,r}$ (defined after (\ref{eq:M_defn}) above), and the channel averaged over the entire round. The former guarantees that
the estimates provided by training are close enough to guarantee
that the rateless code is decodable, and the latter guarantees the
gap between the rates achieved by the scheme and the empirical
mutual information is small.   A \textit{channel estimation error}
$E_1(r)$ occurs for round $r$ if
    \begin{align}
    \left|
        \empMI{P}{\hWempr{M(r)}}
        -
        \empMI{P}{ \frac{1}{r (\chunk - \train)} \sum_{n=1}^{M(r)} \sum_{i \in \icode_{n,r}} W(y | x, z_i)}
    \right|
        >
        \frac{\rategap}{2}
    \label{eq:traincode2}
    \end{align}
or
    \begin{align}
    \left|
        \empMI{P}{\hWempr{M(r)}}
        -
        \empMI{P}{ \frac{1}{r \chunk} \sum_{n=1}^{M(r)} \sum_{i \in \icode_{n,r} \cup \itrain_{n,r}} W(y | x, z_i)}
    \right|
        > \frac{\rategap}{2}~.
        \label{eq:trainwhole2}
    \end{align}

A \textit{decoding error} $E_2(r)$ happens in round $r$ if the
rateless code selected by the encoder and decoder experiences an
error.

\subsection{Preliminaries: Bounding the length of a round}

Before preceding to identify the error events, we will provide
bounds on the length of a round. Our reasons for establishing these are
two-fold. First, if a round fails to terminate or does not result in
successful decoding, the round length should be sufficiently small
so that its impact on the overall rate should be small. Second, when
taking union bounds over chunks in a round, the round length should
be small enough to guarantee the corresponding error probabilities
are small. Moreover, it helps set the maximum length for the
randomized rateless code, defined on page
\pageref{eq:rateless_defn}. Lemma \ref{lemma:M_upperbd} provides
bounds on $M(r)$, the number of chunks in round $r$, which can be
expressed equivalently as $\ell_r/ \chunk(N)$, where $\ell_r$ is
defined in (\ref{eq:roundlength_defn}). For simplicity, we will use
$M$ to denote $M(r)$ when the round $r$ is clear from context.

\begin{figure}
\begin{center}

\psset{xunit=7cm,yunit=8cm}
\begin{pspicture*}(-0.2,0.1)(1.1,0.7)

  \psplot[plotpoints=200,algebraic=true,%
    linewidth=1pt,linecolor=blue]{0.25}{0.9}{(100)^(-x)+0.25}
  \psline[linecolor=blue](0,0.3)(0.9,0.3)
  \psline[linestyle=dashed,linecolor=red](0,0.25)(0.9,0.25)
  \psline[linestyle=dashed](0.65,0.2)(0.65,0.3)
  \rput(-0.1,0.32){$\badrate+\rategap$}
  \rput(-0.05,0.23){$\rategap$}
  \rput(0.65,0.15){$M^{\ast}$}

\psline{->}(0.65,0.55)(0.45,0.40)

\rput(0.65,0.6){empirical rate}

\rput(0,0.65){$\empMI{P}{\hWempr{M}}$} \rput(0.90,0.15){$M$}

\psaxes[linewidth=0.05,arrows=->](0,0.2)(0,0.2)(0.9,0.6)
\end{pspicture*}

\caption{\label{fig:M_finite} Curve of the empirical rate illustrating the bounds on $M$.  The upper bound $\mhi$ is given by \ref{eq:Mstardef}}
\end{center}
\end{figure}
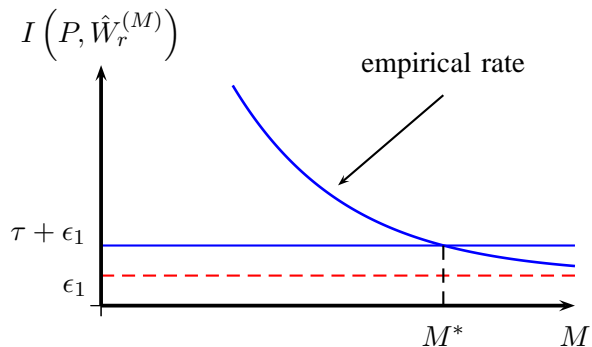

\begin{lemma}[Bounds on $M$]            \label{lemma:M_upperbd}
Fix $\rategap > 0$ and $\badrate > 0$.  Then for the scheme described in Section \ref{subsec:algorithm}, the stopping time $M$ satisfies $M \le \mhi$, where
    \begin{align}
     \mhi := \left\lceil
        \frac{\rbits (N)}{\left(\chunk (N) - \train (N) \right) \cdot
            \badrate}
    \right\rceil~. \label{eq:Mstardef}
    \end{align}
If the decoder attempted to decode, then $M \ge \mlo$, where
    \begin{align}
    \mlo = \frac{\rbits (N)}{\left(\chunk (N) - \train (N) \right) \cdot \capmax} ~.
    \end{align}
\end{lemma}
\begin{proof}
The argument is illustrated in Figure \ref{fig:M_finite}.  The empirical rate given by (\ref{eq:decrule}) is shown in the curve.  The empirical rate $\frac{\rbits}{(\chunk - \train) \times M}$ decreases monotonically with $M$.  In order for the algorithm to continue at time $M$, from (\ref{eq:decrules}) we must have $\frac{\rbits}{(\chunk - \train) \times M} \ge \empMI{P}{\hWempr{n}} - \rategap \ge \badrate$.  Rearranging shows that $M$ must be less than $\mhi$ in (\ref{eq:Mstardef}).  The lower bound is trivial from the definition in (\ref{eq:decrule}) and the cardinality bound on mutual information.
\end{proof}

\subsection{Channel estimation for a single round}

In this section, we provide an upper bound on the error event
$E_1(r)$. The argument relies on the following observation: if
sufficiently many samples are collected to estimate the channel,
these estimates converge to the overall average channel. Lemmas
\ref{lemma:simp_chan_est} and \ref{lemma:chan_est} make this
precise. That is, with a modest number of randomly chosen training
symbols, the decoder can estimate the empirical mutual information
of the channel such that the probability of the channel estimation
error event $E_1(r)$ is small.

\begin{lemma}[Simple channel estimation]  \label{lemma:simp_chan_est}
Recall the chunk training estimates defined in
(\ref{eq:chunk_estimate}), and let parameters satisfy the conditions
in (\ref{eq:param_assump}). Then for any $\chunkeps
> 0$ there exists an $N$ sufficiently
large and constant $a_1$ such that for the $j$-th chunk the training
estimates satisfy:
    \begin{align}
    \prob\left( \left|
        \hWempnl{x}{y}{j}
        -
        W_{\mbf{z}(\icode_{r,j} \cup \itrain_{r,j})}(y|x)
        \right|
        \ge
        \chunkeps \ \forall~x,y
        \right)
        &\le
        \exp \left( -a_1 \chunkeps^2 \train \right) \\
    \prob\left( \left|
        \hWempnl{x}{y}{j}
        -
        W_{\mbf{z}(\icode_{r,j})}(y|x)
        \right|
        \ge
        \chunkeps \ \forall~x,y
        \right)
        &\le
        \exp \left( -a_1 \chunkeps^2 \train \right)~,
    \end{align}
where $t$ is the size of the training set $\itrain_{r, j}$.
\end{lemma}

\begin{proof}
Proving the claim requires two applications of Hoeffding's
inequality \cite{Hoeffding:63inequality} to the training data. The
first uses the sampling with replacement version of the
inequality to show that the training estimates are close to the
state-averaged channel at those training positions. The second uses
the sampling without replacement version to show that the
state-averaged channel in the training positions is close to the
state-averaged channel over the entire chunk. An application of the
triangle inequality and our parameter assumptions in
(\ref{eq:param_assump}) complete the argument.

We now make this precise. First consider the random variables $\{
\mbf{1}(y_i = y) : i \in \itrain_{r,j}(x)\}$ for each $x$ and $y$.
Their expectations over the channel are $\{ W(y | x, z_i) : i \in
\itrain_{r,j}(x)\}$. Applying Hoeffding's inequality to these
variables shows that their mean, which is $\hWempnl{x}{y}{j}$, is
close to $W_{\mbf{z}(\itrain_{r,j}(x))}$, the average channel during
the training:
    \begin{align} \label{eq:SimpChanEst_Hoeffding_with}
    \prob\left( \left|
        \hWempnl{x}{y}{j}
        -
        W_{\mbf{z}(\icode_{r,j} \cup \itrain_{r,j})}(y|x)
        \right|
        \ge
        \traineps
        \right)
        \le
        2 \exp \left( - 2 \frac{\train}{|\mc{X}|} \traineps^2 \right)~.
    \end{align}

Now, recall that the training positions $\itrain_{r,n}$, defined on
page \pageref{sect:training}, are sampled uniformly without
replacement from the whole chunk, so the average channel
$W_{\mbf{z}(\itrain_{r,j}(x)}(y|x)$ is itself a random variable
formed by averaging the random variable $\{ W(y | x, z_i) : i \in
\itrain_{r,j}(x) \}$.  The mean of each of these variables is
$W_{\mbf{z}(\icode_{r,j} \cup \itrain_{r,j})}$, the state averaged
channel over the whole chunk.  For sampling without replacement,
another result of Hoeffding \cite[Theorem 4]{Hoeffding:63inequality}
states that the same exponential inequalities for sampling with
replacement hold, so the channel during the training is a good
approximation to the entire channel during the chunk:
    \begin{align} \label{eq:SimpChanEst_Hoeffding_wo}
    \prob\left( \left|
        W_{\mbf{z}(\itrain_{r,j}(x))}
        -
        W_{\mbf{z}(\icode_{r,j} \cup \itrain_{r,j})}
        \right|
        \ge
        \traineps
        \right)
        \le
        2 \exp \left( - 2 \frac{\train}{|\mc{X}|} \traineps^2 \right)~.
    \end{align}
By applying the triangle inequality to equations
(\ref{eq:SimpChanEst_Hoeffding_with}) and
(\ref{eq:SimpChanEst_Hoeffding_wo}), we have the following:
        \begin{align} \label{eq:SimpChanEst_Hoeffding_combo}
    \prob\left( \left|
        \hWempnl{x}{y}{j}
        -
        W_{\mbf{z}(\icode_{r,j} \cup \itrain_{r,j})}
        \right|
        \ge
        2 \traineps
        \right)
        \le
        4 \exp \left( - 2 \frac{\train}{|\mc{X}|} \traineps^2 \right)~.
    \end{align}
Finally, observe the following:
    \begin{align}
    \left| W_{\mbf{z}(\icode_{r,j} \cup \itrain_{r,j})}
        - W_{\mbf{z}(\icode_{r,j})} \right|
    &= \left| \frac{1}{\chunk} \sum_{i \in \icode_{r,j} \cup \itrain_{r,j}} W(y | x, z_i)
            - \frac{1}{\chunk-\train} \sum_{i \in \icode_{r,j}} W(y | x, z_i) \right| \\
    &= \left| \frac{1}{\chunk} \sum_{i \in \itrain_{r,j}} W(y | x, z_i)
            - \frac{ \train}{ \chunk ( \chunk - \train)}
                    \sum_{i \in \icode_{r,j}} W(y | x, z_i) \right| \\
    &\le 2 \frac{\train}{\chunk} \label{eq:SimpChanEst_params}~.
    \end{align}
The assumptions in (\ref{eq:param_assump}) imply that
(\ref{eq:SimpChanEst_params}) can be made small for sufficiently
large $N$. Thus for $N$ sufficiently large, another application of
the triangle inequality to (\ref{eq:SimpChanEst_Hoeffding_combo})
and (\ref{eq:SimpChanEst_params}) gives the following:
    \begin{align}
    \prob\left( \left|
        \hWempnl{x}{y}{j}
        -
        W_{\mbf{z}(\icode_{r,j})}
        \right|
        \ge
        3 \traineps
        \right)
        \le
        4 \exp \left( - 2 \frac{\train}{|\mc{X}|} \traineps^2 \right)~.
    \end{align}
Choosing $\chunkeps = 3 \traineps$ and a union bound over all $x \in \mc{X}$ and $y \in \mc{Y}$ we get
	\begin{align}
    \prob\left( \left|
        \hWempnl{x}{y}{j}
        -
        W_{\mbf{z}(\icode_{r,j})}
        \right|
        \ge
        \chunkeps
        \right)
        & \le
        \exp \left( - (2/9) \frac{\train}{|\mc{X}|} \chunkeps^2 +  \log |\mc{X}| |\mc{Y}| + \log 4 \right)~. \\
    &\le
    \exp \left( - \alpha_1 \chunkeps^2 \train \right)~,
    \end{align}
where the last inequality follows from taking $N$ sufficiently large and the fact that $\train(N)$ increases with $N$. 
\end{proof}

\begin{lemma}[Channel estimation]           \label{lemma:chan_est}
Recall the error event $E_1(r)$ defined on page
\pageref{sect:error_events}, and let the parameters satisfy the
conditions in (\ref{eq:param_assump}). Then for any $\rategap
> 0$ there exists $N$ sufficiently large and an $a_2 > 0$ such
that for any round $r$ and any state sequence $\mbf{z} \in
\mc{Z}^{M(r) \chunk}$,
    \begin{align}
    \prob\left(
        \left|
        \empMI{P}{\hWempr{M(r)}}
        -
        \empMI{P}{ \frac{1}{r (\chunk - \train)} \sum_{n=1}^{M(r)} \sum_{i \in \icode_{n,r}} W(y | x, z_i)}
    \right|
        >
        \frac{\rategap}{2}
    \right)
        &\le \exp\left( -a_2 \train \right)
    \label{eq:sample_codeMI}
        \\
    \prob\left(
        \left|
        \empMI{P}{\hWempr{M(r)}}
        -
        \empMI{P}{ \frac{1}{r \chunk} \sum_{n=1}^{M(r)} \sum_{i \in \icode_{n,r} \cup \itrain_{n,r}} W(y | x, z_i)}
    \right|
        > \frac{\rategap}{2}
    \right)
        &\le \exp\left( -a_2 \train \right)~.
    \label{eq:sample_roundMI}
    \end{align}
Therefore $\prob(E_1(r)) \le 2 \exp( -a_2 \train)$.
\end{lemma}

\begin{proof}
For all $(x,y)$, Lemma \ref{lemma:simp_chan_est} guarantees that for any $\chunkeps > 0$ the channel estimated during the training of any chunk is within $\chunkeps$ of the average channel during the whole chunk and during the codeword positions with probability $\exp \left( -a_1 \chunkeps^2 \train \right)$.  For a round of length $M(r)$, a union bound over chunks shows that
    \begin{align}
    \prob\left( \left|
        \frac{1}{M(r)} \sum_{j=1}^{M(r)} \hWempnl{x}{y}{j}
        -
        \frac{1}{M(r)} \sum_{j=1}^{M(r)} W_{\mbf{z}(\icode_{r,j} \cup \itrain_{r,j})}(y|x)
        \right|
        \ge
        \chunkeps \ \forall~x,y
        \right)
        &\le
        M(r) \exp \left( -a_1 \chunkeps^2 \train \right)
        \label{eq:sample_round} \\
    \prob\left( \left|
        \frac{1}{M(r)} \sum_{j=1}^{M(r)} \hWempnl{x}{y}{j}
        -
        \frac{1}{M(r)} \sum_{j=1}^{M(r)} W_{\mbf{z}(\icode_{r,j})}(y|x)
        \right|
        \ge
        \chunkeps \ \forall~x,y
        \right)
        &\le
         M(r)  \exp \left( -a_1 \chunkeps^2 \train \right)~.
        \label{eq:sample_code}
    \end{align}
Since $M(r)$ is at most $\mhi$, for $N$ sufficiently large the effect of the union bound is negligible.

The remainder of the proof is to show that if the channel estimated from the training is close with high probablity to both the average channel during the codeword positions and the average channel during the whole round, then the empirical mutual informations must be close as well.  Lemma \ref{lemma:var_bnd_mi} in the Appendix shows exactly this.  For any $\rategap > 0$ there exists a $\chunkeps > 0$ and $N$ sufficiently large such if the events in (\ref{eq:sample_round}) and (\ref{eq:sample_code}) fail to hold then the events in (\ref{eq:sample_roundMI}) and (\ref{eq:sample_codeMI}) also fail to hold.  This completes the proof.
\end{proof}

\textbf{Remark:} Under the parameter assumptions in equation (\ref{eq:param_assump}), the number of bits of common randomness needed in Lemmas \ref{lemma:simp_chan_est} and \ref{lemma:chan_est} to specify the training positions is sublinear in the blocklength $N$. Note that a similar conclusion was reached 
by Shayevitz and Feder for their scheme, which also uses training positions to the estimate the channel 
 \cite{ShayevitzF:06empirical}. This point is discussed in more detail in Section \ref{sec:discussion} 
 on page \pageref{discussion:commonrandomness}.

\subsection{Rateless coding}

The last ingredient in our strategy is the rateless code used during
each round.  The key property we need is that if the empirical rate
drops below the empirical mutual information of the channel, then the code
can be decoded with small probability of error.

\begin{lemma}[Rateless codes]  \label{lemma:rateless}
For any $\delta' > 0$ and distribution $P$, there exists an integer $\codechunk$ sufficiently large, $\rlesseps > 0$ and an $(\mhi, \codechunk, \rbits)$ randomized rateless code
defined in Section~\ref{sect:proposed_strategy}
such that if at decoding time $M$ the state sequence $\mbf{z}_{1}^{M \codechunk}$ satisfies
	\begin{align}
	\frac{k}{M \codechunk} \le \empMI{P}{W_{\mbf{z}_{1}^{M \codechunk}}} -  \delta'~,
	\end{align}
then its maximal error $\ratelesserr(M,\mbf{z})$, defined in (\ref{eq:randomized_rateless_defn}), satisfies
	\begin{align}
	\ratelesserr(M,\mbf{z}) < \exp( - M \codechunk \rlesseps )~.
	\end{align}
\end{lemma}

\begin{proof}
Fix $\delta'$ and a distribution $P$. We can approximate $P$ arbitrarily closely with a type of a sufficiently large denominator, so without loss of generality, we assume $P$ is a type and choose $\codechunk$ to be large enough so that the denominator of type $P$ divides $\codechunk$.   Let $\mc{C}_M(J)$ be a randomized rateless code. Specifically, $\mc{C}_M(J)$ is a random variable distributed on the set of rateless codes of blocklength $M \codechunk$ whose $J$ codewords are drawn independently and uniformly from the composition-$P$ set $\typ{M\codechunk}{P}$ and with a maximum mutual information (MMI) decoder. The remainder of the proof can be sketched as follows: we verify that 
the codebook $\mc{C}_M(J)$ has satisfactory error performance under the assumptions of this Lemma. Then, we construct a codebook $\mc{D}_M(K)$ by keeping only those codewords in $\mc{C}_M(J)$ whose composition is $P$ in each chunk of $\codechunk$ symbols.  We then show that the distribution of $\mc{D}_M(K)$ is the same as that of a codebook $\mc{E}_{\mhi}(K)$ truncated to blocklength $M \codechunk$.

\textbf{Codebook properties.}  Before proceeding to construct $\mc{D}_M(K)$, we first examine properties of the constant-composition codebook $\mc{C}_M(J)$ of composition $P$. Recall the definition of maximal error for randomized rateless codes in  (\ref{eq:randomized_rateless_defn}) and (\ref{eq:randomized_rateless_msgerr_defn}). A result of Hughes and Thomas \cite[Theorem 1]{HughesT:96exponent} shows that for sufficiently large $M\codechunk$, there exists a function $E_r$ such that for all $J > 0$, $\delta > 0$, and distribution $Q$ on $\mc{Z}$,
	\begin{align}
	\max_{\mbf{z} \in \typ{M \codechunk}{Q}} \max_{j \in [J]} \varepsilon_j(M, \mbf{z}, \mc{C}_M(J))
		&\le \exp \left( - M \codechunk \left[ E_r((M \codechunk)^{-1} \log J + \delta,W,P,Q) - \delta \right] \right)
		\label{eq:ht_bound} \\
		E_r((M \codechunk)^{-1} \log J + \delta,W,P,Q) &\geq \max \left\{ 0, ~\empMI{P}{W_Q} - \delta -\frac{1}{M \codechunk} \log J \right\} ~.
	\end{align}
Fix $\cbookeps = \frac{\delta'}{4}$ and let $\mc{Q}(M)$ be the set of all $Q$ such that 
	\begin{align}
	0 < \frac{\delta'}{4} \leq \empMI{P}{W_Q} - 2 \delta - \frac{1}{M \codechunk} \log J ~.
	\label{eq:avcratebound}
	\end{align}
If $Q \in \mc{Q}(M)$, then we can rewrite the bound in (\ref{eq:ht_bound}) as follows:
	\begin{align}
	\max_{\mbf{z} \in \typ{Mn}{Q} : Q \in \mc{Q}(M)} \max_{j \in [J]} \varepsilon_j(M, \mbf{z}, \mc{C}_M(J))
	\le 
	\exp \left( - M \codechunk \cbookeps \right)~.
	\end{align}
In particular, this gives the following bound on the expectation over $\mc{C}_M(J)$ of the \textit{average error}:
	\begin{align}
	\max_{\mbf{z} \in \typ{Mn}{Q} : Q \in \mc{Q}(M)} 
		\expe_{\mc{C}_M(J)} \left[
			\frac{1}{J} \sum_{j=1}^{J}  \varepsilon_j(M, \mbf{z}, \mc{C}_M(J))
		\right]
	\le
	\exp \left( - M \codechunk \cbookeps \right)~.
	\end{align}

Use Markov's inequality to bound the probability that the average error exceeds a given value $\alpha_1$:
	\begin{align}
	\max_{\mbf{z} \in \typ{Mn}{Q} : Q \in \mc{Q}(M)} 
	\prob_{\mc{C}_M(J)}\left( \frac{1}{J} \sum_{j=1}^{J} \varepsilon_j(M, \mbf{z}, \mc{C}_M(J)) 
	\ge \alpha_1(\codechunk ,M)
	\right)
	\le 
	\frac{ \exp \left( - M \codechunk \cbookeps \right) }{ \alpha_1(\codechunk ,M) }~.
	\end{align}
This establishes that for any $\delta > 0$ the codebook has average error no more than $\alpha_1(M)$ with high probability.  

\textbf{Expurgation.}  We define a thinning operation on the codebook $\mc{C}_M (J)$ to form 
the codebook $\mc{D}_M(K)$ as follows: remove all codewords in $\mc{C}_M(J)$ which are not in the piecewise constant-composition set $\{\typ{\codechunk}{P}\}^M$.  That is, we keep only those codewords which have type $P$ in each chunk.  If there are fewer than $K$ remaining codewords after this expurgation, declare an encoding error -- if there are more than $K$ then keep the first $K$ codewords.  The decoding rule is the same MMI rule as before.

The probability of this encoding error can be bounded using Lemma \ref{lemma:cat_set_bd} on page \pageref{lemma:cat_set_bd}, which states that the probability that a codeword drawn uniformly from $\typ{M \codechunk}{P}$ is also in the set $\{\typ{\codechunk}{P}\}^M$ is at least $\beta_0(\codechunk,M) = \exp(- \eta M \log(\codechunk+1) )$ for $\codechunk$ sufficiently large.  Therefore the expected number of codewords in $\mc{C}_M(J)$ that survive the thinning is at least $J \exp(- \eta M \log \codechunk )$.  Since the codewords are i.i.d., the probability that the number of codewords surviving the thinning is at least $\beta J$ can be bounded:	
	\begin{align}
	\prob\left( \left| \mc{C}_M(J) \cap \{\typ{\codechunk}{P}\}^M \right| \le \beta J \right)
		\le
		J \cdot \exp\left( - J \cdot D\left( \beta \big\| \beta_0(c,M) \right) \right)~.
	\end{align}
By choosing $K = \beta_0(\codechunk,M)^2 J$, which corresponds to $\beta = \beta_0(\codechunk,M)^2$, the probability of encoder error can be made arbitrarily small.  The rate of codebook $\mc{D}_M(K)$ is
	\begin{align}
	\frac{1}{M \codechunk} \log K = \frac{1}{M \codechunk} \log J - \frac{2 \eta \log \codechunk }{\codechunk}~.
	\end{align}
Setting $k = \log K$, note from (\ref{eq:avcratebound}), for sufficiently large $\codechunk$ the error can be made small as long as 
	\begin{align}
	\frac{k}{M \codechunk} \le \min_{Q \in \mc{Q}(M)} \empMI{P}{W_Q} - 3 \delta - \frac{\delta'}{4}~.
	\label{eq:thinratebound}
	\end{align}
Setting $\delta = \delta'/4$ in the original construction of $\mc{C}_M(J)$, for sufficiently large $\codechunk$, equation (\ref{eq:thinratebound}) guarantees a bound on the error.  In particular, since the codewords of $\mc{D}_M(K)$ are a subset of the codewords of $\mc{C}_M(K)$, the average error can increase at most by a factor of $J/K$:
	\begin{align}
	\max_{\mbf{z} \in \typ{Mn}{Q} : Q \in \mc{Q}(M)} 
	\prob_{\mc{C}_M(J)}\left( \frac{1}{K} \sum_{j=1}^{K} \varepsilon_j(M, \mbf{z}, \mc{D}_M(K)) 
	\ge \frac{\alpha_1(\codechunk, M)}{\beta_0(\codechunk,M)^2}
	\right)
	\le 
	\frac{ \exp \left( - M \codechunk \cbookeps \right) }{ \alpha_1(\codechunk, M)}~.
	\label{eq:thinerrorbound}
	\end{align}
This shows that for any $\delta' > 0$ the average error can be bounded.

\textbf{Nesting.}  Consider the codebook $\mc{E}_M(K)$ formed by drawing $K$ codewords independently uniformly distributed on $\{\typ{\codechunk}{P}\}^M$ together with the MMI decoding rule.  It is clear that $\mc{D}_M(K)$ has the same distribution as $\mc{E}_M(K)$, so the bound (\ref{eq:thinerrorbound}) holds for $\mc{E}_M(K)$ as well:
	\begin{align}
	\max_{\mbf{z} \in \typ{Mn}{Q} : Q \in \mc{Q}(M)} 
	\prob_{\mc{E}_M(K)}\left( \frac{1}{K} \sum_{j=1}^{K} \varepsilon_j(M, \mbf{z}, \mc{E}_M(K)) 
	\ge \frac{\alpha_1(\codechunk, M)}{\beta_0(\codechunk, M)^2}
	\right)
	\le 
	\frac{ \exp \left( - M \codechunk \cbookeps \right) }{ \alpha_1(\codechunk, M)}~.
	\label{eq:nesterrorbound}
	\end{align}

Note that $\mc{E}_M(K)$ has the same distribution as the codebook $\mc{E}_{\mhi}(K)$ truncated to blocklength $M \codechunk$.  The set of $\mbf{z} \in \mc{Z}^{\mhi \codechunk}$ for which the bounds (\ref{eq:nesterrorbound}) hold is
	\begin{align}
	\mc{Z}(K) = \left\{
		\mbf{z} \in \mc{Z}^{\mhi \codechunk}
		:
		(z_1, \ldots, z_{M\codechunk}) \in \typ{M \codechunk}{Q},\ 
		Q \in \mc{Q}(M),\ 
		M \in \{\mlo, \ldots, \mhi\}
		\right\}~.
	\end{align}
For any $\mbf{z}$ in this set and decoding time $M$ such that $(z_1, \ldots, z_{M\codechunk}) \in \typ{M \codechunk}{Q}$ for some $Q \in \mc{Q}(M)$, the probability that the random codebook $\mc{E}_{\mhi}(K)$ truncated to blocklength $M$ has average error probability exceeding $\frac{\alpha_1(\codechunk, M)}{\beta_0(\codechunk, M)^2}$ can be made arbitrarily small.

\textbf{Back to maximal error.}  The equation (\ref{eq:nesterrorbound}) says that the average error under the randomized code $\mc{E}_M(K)$ can be made arbitrarily small.  Standard results on AVCs \cite[Exercise 2.6.5]{CsiszarKorner:82book} show that by permuting the message index the same bound holds for the maximal error.  Thus with probability $1 - \exp(- M \codechunk \cbookeps)/\alpha_1(\codechunk, M)$ the randomly selected codebook has maximal error smaller than $\frac{\alpha_1(\codechunk, M)}{\beta_0(\codechunk, M)^2}$.  The probability of encoding error is vanishingly small with respect to these quantities, so the total probability of error can be upper bounded:
	\begin{align}
	\ratelesserr(M,\mbf{z})
		&<
		\max\left( 
			\frac{ \exp \left( - M \codechunk \cbookeps \right) }{ \alpha_1(\codechunk, M)}, 
			\frac{\alpha_1(\codechunk, M)}{\beta_0(\codechunk, M)^2}		\right) \\
		&<
		\max\left( 
			\frac{ \exp \left( - M \codechunk \cbookeps \right) }{ \alpha_1(\codechunk, M)}, 
			\frac{\alpha_1(\codechunk, M)}{\exp(-2 \eta M \log \codechunk)}		
			\right)~.
	\end{align}
Selecting $\alpha_1(\codechunk, M) = \exp( - M \codechunk \cbookeps/2 )$ yields the following bound
for sufficiently large $\codechunk$:
	\begin{align}
	\ratelesserr(M,\mbf{z}) < \exp( - M \codechunk \cbookeps/3 )~.
	\end{align}
Setting $\rlesseps = \cbookeps/3$ yields the result.
\end{proof}

\textbf{Remark:}  As stated, the codebook constructed in Lemma \ref{lemma:rateless} requires a very large amount of common randomness shared between the encoder and decoder.  
This issue is discussed in more detail in Section \ref{sec:discussion} 
 on page \pageref{discussion:commonrandomness}.

\subsection{Proof of Theorem \ref{theorem:main_result}}
We now combine the results in the previous sections to prove Theorem
\ref{theorem:main_result}. Namely, in Section
\ref{sect:error_events}, we defined error events $E_1(r)$ and
$E_2(r)$. We then provided bounds on $E_1(r)$ in Lemma
\ref{lemma:chan_est} and proved the existence of a randomized
rateless code with a small maximal error probability in Lemma
\ref{lemma:rateless}. As will be seen in the proof, Lemmas \ref{lemma:chan_est}
and \ref{lemma:rateless} provide a bound on $E_2(r)$. By combining this bound with the bound 
on $E_1(r)$ and parameter assumptions in (\ref{eq:param_assump}), the result follows straightforwardly.

\begin{proof} The proof is divided into three parts. We first establish in equation (\ref{eq:fb_small}) that for sufficiently large $N$, the feedback rate can be made arbitrarily small. In the second part, we bound the error probability in (\ref{eq:algerror_bd_combo}). In the third part, we give a lower bound on the rate under the assumption the error event does not occur, which leads to equation (\ref{eq:rateloss_small}). These parts establish all necessary components in the statement of the result.

We use the coding strategy proposed in Section \ref{sect:proposed_strategy}. Note that under 
the parameter assumptions in (\ref{eq:param_assump}), for all $\fbrate^\ast > 0$, there exists sufficiently large $N$ such that the feedback rate (\ref{eq:strategy_fbrate}) satisfies the following bound:
	\begin{align} \label{eq:fb_small}
		\Rfb < \fbrate^\ast~.
	\end{align}
Fix a sequence $\mbf{z}$.  The scheme induces a partition of $\mbf{z}$ into rounds $r = 1, 2, \ldots$ at times $\{\ell_r\}$.  Let $\mbf{z}(r) = \mbf{z}_{\ell_{r-1} + 1}^{\ell_r}$ be the
state sequence during the $r$-th round.  The type of $\mbf{z}$ can be written as:
    \begin{align}
    \styp{\mbf{z}} = \sum_{r} \frac{\ell_{r} - \ell_{r-1}}{N} \styp{\mbf{z}(r)}~,
    \end{align}
where $\ell_{r}$ is the length of a round, as defined in equation (\ref{eq:roundlength_defn}). 
Lemma \ref{lemma:chan_est} shows that for any $\rategap > 0$ there exists an $N$ sufficiently large such that the channel estimation error probability $\prob(E_1(r))$ is exponentially small.    Taking a union bound over all rounds, the probability of estimation error is
    \begin{align} \label{eq:algerror_bd0}
    \prob\left( \bigcup_{r} E_1(r) \right)
        \le 2 \frac{N}{\chunk} \exp\left( -a_2 \train \right)~.
       \end{align}
By the parameter assumptions in (\ref{eq:param_assump}), $N/\chunk$ and $\train$ grow polynomially in $N$, so for large $N$ the exponential term dominates and the probability of an estimation error in any round goes to $0$.  Given any $\algerror > 0$, for sufficiently large $N$, equation (\ref{eq:algerror_bd0}) gives the following bound:
	\begin{align} \label{eq:algerror_bd1}
		    \prob\left( \bigcup_{r} E_1(r) \right)
        \le \frac{\algerror}{2}~.
	\end{align}
	
Suppose round $r$ was terminated due to ``BAD NOISE.''  In this case, from (\ref{eq:badnoise_condition}) we have the following:
    \begin{align}
    \empMI{P}{\hWempr{M(r)}} - \rategap < \badrate~.
    \end{align}
By Lemma \ref{lemma:chan_est}, $\empMI{P}{\hWempr{M(r)}}$ is close to $\empMI{P}{W_{\mbf{z}(r)}}$.  That is, there exists an $N$ sufficiently large such that with probability $1 - \exp\left( -a_2 \train \right)$, we have that $\empMI{P}{W_{\mbf{z}(r)}} < \badrate + 3 \rategap/2$.
For any $\rateloss > 0$, we can choose a large $N$ and small $\badrate$ such that the following holds for all ``BAD NOISE'' rounds:
    \begin{align} \label{eq:rateloss_badnoise}
    \empMI{P}{W_{\mbf{z}(r)}} < \rateloss/2~.
    \end{align}
Therefore, for rounds which are terminated due to bad noise, the state sequence $\mbf{z}(r)$ has a type $\styp{\mbf{z}(r)}$ such that $\empMI{P}{W_{\mbf{z}(r)}}$ is small.

Now suppose the decoder attempted to decode at the end of round $r$.  Then (\ref{eq:decrule}) implies that the estimated empirical mutual information from the training satisfies a different inequality:
        \begin{align}
         \empMI{P}{\hWempr{M(r)}} - \rategap & > \frac{\rbits}{(\chunk - \train) \cdot M(r)}~.
       \end{align}
If the event $E_1(r)$ does not happen, then $\empMI{P}{\hWempr{M(r)}}$ is within $\rategap/2$ of the empirical mutual information during the non-training positions:
    \begin{align} \label{eq:emp_mi_est}
    \frac{\rbits}{(\chunk - \train) \cdot M(r)}
    <
    \empMI{P}{ \frac{1}{r (\chunk - \train)} \sum_{n=1}^{M(r)} \sum_{i \in \icode_{n,r}} W(y | x, z_i)}
            - \frac{\rategap}{2}~.
    \end{align}
Thus, conditioned on $E_1^c(r)$ and under our assumption (\ref{eq:param_assump}), (\ref{eq:emp_mi_est}) and Lemma \ref{lemma:rateless} imply that for $\delta' = \rategap/2$ there exists a sufficiently large $N$, exponent $\rlesseps > 0$, and an $(\mhi, \chunk - \train, \rbits)$ randomized rateless code with error
$\ratelesserr(M,\mbf{z}) < \exp( - M (\chunk-\train) \rlesseps )$ for every round $r$ in which decoding occurs. A union bound then implies the decoding error probability over all rounds in which decoding occurs can be bounded:
    \begin{align} \label{eq:algerror_bd2}
    	\prob\left( \bigcup_{r} E_2(r) ~ \left| ~ \bigcap_{r} E_1^c (r) \right. \right) \leq \frac{N}{\chunk} \exp( - (\chunk-\train) \rlesseps )~.
    \end{align}
By (\ref{eq:param_assump}), this can be made arbitrarily small for sufficiently large $N$, and therefore for any $\algerror > 0$, (\ref{eq:algerror_bd1}) and (\ref{eq:algerror_bd2}) imply there exists an $N$ sufficiently large such that the estimation error and decoding error can be made smaller than $\algerror$:
    \begin{align} \label{eq:algerror_bd_combo}
    	\prob\left( \bigcup_{r, i=1,2} E_i(r) \right) \leq \algerror~.
    \end{align}

The remaining thing is to calculate the rate, given that none of the error events occur. If the decoder attempted to decode after $M(r)$ chunks, then after $M(r)-1$ chunks the threshold condition in (\ref{eq:decrule}) was not satisfied:
    \begin{align}
    \frac{\rbits}{(\chunk - \train) \cdot (M(r)-1)}
         \ge
         \empMI{P}{\hWempr{M(r)-1}} - \rategap~,
    \end{align}
Our assumption in equation (\ref{eq:param_assump}) that $(\chunk(N))^2/\rbits(N) \to 0$ and our lower bound on the length of a round in Lemma \ref{lemma:M_upperbd} is $\Theta(\rbits(N)/\chunk(N))$ channel uses imply that for sufficiently large $N$, the amount that the estimated mutual information can change over the course of a single chunk ($\chunk(N)$ channel uses) can be made arbitrarily small.  More formally, for any $\chunkmoveeps > 0$, for sufficiently large $N$,
    \begin{align}
    \left| \empMI{P}{\hWempr{M(r)-1}} - \empMI{P}{\hWempr{M(r)}} \right| < \chunkmoveeps~.
    \end{align}
Thus
    \begin{align}
    \frac{\rbits}{(\chunk - \train) \cdot M(r)}
    &= \left( 1 - \frac{1}{M(r)} \right)
        \frac{\rbits}{(\chunk - \train) \cdot (M(r)-1)} \\
    &\ge
        \left( 1 - \frac{1}{M(r)} \right)
        \left( \empMI{P}{\hWempr{M(r)-1}} - \rategap \right) \\
    &\ge
        \left( 1 - \frac{1}{M(r)} \right)
        \left( \empMI{P}{\hWempr{M(r)}} - \chunkmoveeps - \rategap \right)~.
    \end{align}
Finally, the overall empirical rate for the round is slightly lower because of overhead from training:
    \begin{align}
    \frac{\rbits}{\chunk M(r)}
    \ge
        \left( 1 - \frac{1}{M(r)} \right)
        \left( 1 - \frac{\train}{\chunk} \right)
        \left( \empMI{P}{\hWempr{M(r)}} - \chunkmoveeps - \rategap \right)
    \end{align}
Under the assumptions in (\ref{eq:param_assump}) and conditioned on (\ref{eq:trainwhole2}) not occurring, for any $\rateloss > 0$ there exists an $N$ sufficiently large such that
    \begin{align} \label{eq:rateloss_decoding}
    \frac{\rbits}{\chunk M(r)} \ge \empMI{P}{\hWempr{M(r)}} - \rateloss/2~.
    \end{align}

The final thing to consider is the last round $r^{\ast}$ in which the decoder does not decode.  The maximum length of this round is $\mhi \chunk$, and
    \begin{align} 	\label{eq:rateloss_bd0}
    \frac{\ell_{r^{\ast}} - \ell_{r^{\ast}-1}}{N} \empMI{P}{W_{\mbf{z}(r^{\ast})}}
    \le \frac{\mhi \chunk}{N} \max\{|\mc{X}|, |\mc{Y}| \}~.
    \end{align}
By (\ref{eq:param_assump}), for sufficiently large $N$, (\ref{eq:rateloss_bd0}) can be made to satisfy the following condition:
	\begin{align}
		    \frac{\ell_{r^{\ast}} - \ell_{r^{\ast}-1}}{N} \empMI{P}{W_{\mbf{z}(r^{\ast})}} \le \rateloss/2 ~.\label{eq:rateloss_bd1}
	\end{align}
To summarize, for sufficiently large $N$ and each round $r$ in which the decoder feeds back ``BAD NOISE'' or ``DECODED'', the rate at which the scheme decodes can be lower bounded by
    \begin{align} \label{eq:rateloss_bd2}
    R(r) \ge \empMI{P}{W_{\mbf{z}(r)}} - \rateloss/2~,
    \end{align}
which follows from (\ref{eq:rateloss_badnoise}) and (\ref{eq:rateloss_decoding}). Finally, we use  (\ref{eq:rateloss_bd1}), (\ref{eq:rateloss_bd2}), and the convexity of mutual information to provide a lower bound on the overall rate of the scheme:
    \begin{align}
    R
    &\ge \sum_{r=1}^{r^{\ast}-1} \frac{\ell_{r} - \ell_{r-1}}{N} \left(\empMI{P}{W_{\mbf{z}(r)}} - \rateloss/2
            \right) \\
    &\ge \empMI{P}{ \sum_{r} \frac{\ell_{r} - \ell_{r-1}}{N} W_{\mbf{z}(r)} } - \rateloss \\
    &= \empMI{P}{ W_{\mbf{z}} } - \rateloss~. \label{eq:rateloss_small}
    \end{align}
As mentioned above, the result now follows immediately from (\ref{eq:fb_small}), (\ref{eq:algerror_bd_combo}), and (\ref{eq:rateloss_small}).
\end{proof}

\section{Discussion \label{sec:discussion}}

The central question we tried to address in this paper was how much feedback is needed to achieve the channel mutual information in an individual sequence setting of \cite{ShayevitzF:06empirical}.  Limited feedback in two-way and relaying systems have been studied before \cite{Schwartz:63feedback,Hayes:68feedback,AhmedKSA:06outage} and are used in many modern-day communication protocols for control information.   Research interest on limited feedback for multiuser and multiantenna models has grown tremendously (see \cite{LoveHLGRA:08limited} and references therein).  Quantifying the role and possible benefits of limited feedback is an important step in understanding how to structure adaptive communication systems.

In this paper we described a coding strategy under a general channel uncertainty model that uses limited feedback to achieve rates arbitrarily close to an i.i.d. discrete memoryless channel with the same first-order statistics.  Feedback allows the system to adapt the coding rate based on the channel conditions.  When each element in the class of channels over which we are uncertain has the same capacity achieving input distribution, the coding strategy achieves rates at least as large as the empirical capacity, which is defined as the capacity of an i.i.d. discrete memoryless channel with the same first-order statistics. Since the rates that we can guarantee for our scheme are close to the average channel in a round, our total rate over many rounds may in fact exceed the empirical capacity.  This is due to the convexity of mutual information in the channel.

The work is a commentary on an earlier investigation by Shayevitz and Feder \cite{ShayevitzF:06empirical} that considered the case in which the encoder has access to full output feedback from the decoder and  allows the encoder to provide control and estimation information in a set of training sequences that can be selected via common randomness.  Furthermore, their scheme does not require a fixed blocklength in advance and hence has an infinite horizon.  By contrast, our strategy can be viewed as a kind of incremental redundancy hybrid ARQ \cite{Soljanin:03harq}, in which the decoder uses the feedback link to terminate rounds that are too noisy while less noisy rounds are individually decoded.  In order to set the parameters for our scheme we must fix a total blocklength in advance, although it may be possible to redefine the scheme to operate without a horizon, as in \cite{ShayevitzF:06empirical}.  

An interesting point is that our basic algorithm uses standard ``tricks'' for communication systems, such as channel estimation via pilot signals, ARQ with rateless codes, and randomization.  By adapting or reusing technologies that have already been developed, these gains can be realized more easily.  Several open questions and extensions of the algorithm presented here would be of interest, two of which are the following:
\begin{enumerate}
 \item \label{discussion:commonrandomness} {\em The necessary amount of common randomness.} The algorithm presented here requires common randomness between the encoder and the decoder 
to show that zero-rate feedback is sufficient to achieve the empirical mutual information. We now 
provide an account of how much common randomness is required. There are two places where our algorithm requires common randomness, namely, {\em (i),} the selection of the channel training positions, and {\em (ii),} the random selection of the codebook for each round.

For {\em (i),} the training positions, under our parameter assumptions in \eqref{eq:param_assump}, $\log N$ bits are required to indicate the 
position of each of the $\train = \Theta (N^{g_3})$ training positions for each chunk of length $\chunk = \Theta (N^{g_2})$, where $\frac{1}{2} > g_2 > g_3 > 0$.  Since there are $N/\chunk$ chunks, this requires at total of 
	\begin{align*}
		 N \cdot \frac{\train(N)}{\chunk(N)} \cdot \log N = \Theta \left( N^{1-(g_2 - g_3)} \cdot \log N \right) \text{ bits}~,
	\end{align*}
which, under our parameter assumptions is sublinear in $N$. For {\em (ii),} the selection 
of a codebook for each round can require as much as $\mhi \cdot \capmax$ bits of common randomness per codeword for a total of $\mhi \cdot \capmax \cdot 2^{\mhi \cdot \capmax}$ bits of common randomness, where $\capmax = \log \min\{|\mc{X}|,|\mc{Y}|\}$. The total number of rounds can be as large as $\frac{N}{\mlo}$, where $\mhi$ and $\mlo$ are defined in Lemma \ref{lemma:M_upperbd}. Thus, codebook selection requires 
	\begin{align*}
		\mhi \cdot \capmax \cdot 2^{\mhi \cdot \capmax} \cdot \frac{N}{\mlo} = \frac{(\capmax)^2}{\badrate} \cdot N \cdot 2^{\mhi \cdot \capmax} \text{ bits}~,
	\end{align*}
where $\badrate$, defined in \eqref{eq:badnoise_condition}, is a parameter of the algorithm that does not depend on $N$. Thus, the total 
common randomness required is superlinear in $N$.

Reducing common randomness is outside the scope of the current work. However, if common randomness were not available between the encoder and decoder, it could be provided by the feedback link, but then the strategy considered in this paper would require a prohibitively large feedback rate that would increase with the blocklength $N$. To show instead that the feedback rate could be made asymptotically negligible in such a setting, one would need to prove the existence of a strategy for which the total bits of common randomness required would be sublinear in the blocklength $N$.
 
 A potential technique that might be useful could be to adapt tools from the theory of arbitrarily varying channels \cite{Ahlswede:78elimination} to find nested code constructions that use a limited amount of common randomness \cite{SarwateG:07rateless}.  Such an argument would require showing that a randomized code with support on $T = (\mhi \chunk)^2$ codes can be made from iid sampling of the randomized code of Lemma \ref{lemma:rateless}.  This new randomized code could then be used to establish a sublinear number of bits. Specifically, in each round, this new randomized code could be used by selecting one of the $T$ codes for use.  This would require $\log T = O( \log N )$ bits per round for a total cost of at most $O( (N/\mlo) \log N )$, which would be sublinear in $N$.  

Another potential method, more in the interactive coding spirit of feedback systems, could be to show the existence of deterministic list-decodable codes with small list sizes.  If the list is of size $L$, the decoder could find $L$ bits in the message, which could be used to disambiguate the list \cite{Sahai:07balancing}.  By using $L \log k$ bits in the feedback, the decoder could request those $L$ bits from the encoder.  By sacrificing just $O(L)$ more forward channel uses, the encoder could send the $L$ bits with negligible impact to the rate. If the empirical mutual information in the next round were above $\badrate$, this would be 
sufficient for success.  %

\item {\em Adaptation of the channel input, and thus, codebook distribution.}
An apparent limitation of the algorithm presented here is that the channel input distribution is selected once and kept fixed throughout, irrespective of the behavior of the state sequence.  Adaptation of the channel input distribution may lead to higher \textit{or} lower rates.  One interesting question would be whether universal prediction techniques \cite{MerhavF:98universal} can be used in conjunction with channel coding to adapt the channel input.  Another set of interesting questions emerges if we consider performance on a sequence that comes from a certain class of sequences. For example, if one 
were to consider an alternate notion of empirical capacity in which the empirical sequences were estimated as finite-order Markov models, adapting the channel input distribution may give quantifiable benefits.

\end{enumerate}

The individual sequence model considered in this paper is by no means the only way of modeling channel uncertainty.  One model which does away with modeling the channel state was proposed by Lomnitz and Feder \cite{LomnitzF:09isit}.  An alternative model within the state sequence framework is a class of noise models that varies in a piecewise-constant fashion.  This model is related to the on-line estimation problems studied by Kozat and Singer \cite{KozatS:06ita} and may be useful to understand block fading.  For such models we could consider modifying our strategy to adapt the value of $\rbits$ by trying to learn the coherence time of the channel.  In the sense of competitive optimality, the competition class could be coding strategies that know the coherence intervals exactly.
Variations on the model of the feedback link may also lead to interesting new results. Alternative channel models in which the feedback is noisy or allowed to have time-varying rate may present new issues to consider, particularly for the case in which there is uncertainty in the feedback link as well.   For future communications systems that must share common resources, such investigations may shed new light on strategies in these settings.

\section*{Acknowledgments}

We thank Ofer Shayevitz and Meir Feder for providing a preprint of their paper after their presentation of it at the Kailath Colloquium \cite{Feder:06kailath}. This work grew out of a presentation of that work for UC Berkeley's advanced information theory course EE290S. Special thanks go to the other students in that class for helpful discussions.

\appendix

We provide here the proofs of the lemmas used in the analysis of our algorithm\footnote{We were unable to find a standard reference for the entropy bounds below, which is why we provide the derivation.  The proofs can be omitted for space if the reviewers and editors think it appropriate to do so.}.

\subsection{Bounds on entropy and mutual information}

We need a short technical lemma about concave functions.

\begin{lemma}  \label{lemma:concaveLem}
Let $f$ be a concave increasing function on $[a,b]$.  Then if $a \le x \le x + \epsilon \le b$, we have	
	\begin{align}
	f(x+\epsilon) - f(x) \le f(a + \epsilon) - f(a)~. 
	\label{eq:concIneq}
	\end{align}
\end{lemma}

\begin{proof}
Without loss of generality we can take $a = 0$, $b = 1$, and 
$f(a) = 0$.  Now consider
	\begin{align*}
	f(x) = f\left( 
		\frac{x}{x + \epsilon} \cdot (x + \epsilon) 
		+ \frac{\epsilon}{x + \epsilon} \cdot 0 \right) 
		&\ge \frac{x}{x + \epsilon} f(x + \epsilon) 
			+ \frac{\epsilon}{x + \epsilon} f(0) \\
	&= \frac{x}{x + \epsilon} f(x + \epsilon) \\
	f(\epsilon) = f\left( 
		\frac{x}{x + \epsilon} \cdot 0 
		+ \frac{\epsilon}{x + \epsilon} \cdot (x + \epsilon) 
		\right) 
		&\ge \frac{x}{x + \epsilon} f(0) 
		+ \frac{\epsilon}{x + \epsilon} f(x + \epsilon) \\
	&= \frac{\epsilon}{x + \epsilon} f(x + \epsilon)~.
	\end{align*}
Therefore
	\begin{align}
	f(x) + f(\epsilon) \ge f(x + \epsilon)~,
	\end{align}
as desired.
\end{proof}

Using the preceding lemma, we can show that a bound on the total variational distance between two distributions gives a bound on the
entropy between those two distributions.

\begin{lemma} \label{lemma:var_bnd_ent} 
Let $P$ and $Q$ be two distributions on a finite set $\mc{S}$ with $|\mc{S}| \ge 2$.  If
	\begin{align}
		\left| P(s) - Q(s) \right| \le \epsilon \qquad \forall s \in \mc{S}~,
	\end{align}
then 
	\begin{align}
		\left| H(P) - H(Q) \right| \leq (|\mc{S}| - 1) \cdot h_b(\epsilon) + (|\mc{S}| - 1) \log(|\mc{S}| - 1) \cdot \epsilon~,
	\end{align}
where $h_b(\cdot)$ is the binary entropy function.
\end{lemma}

\begin{proof} 
Let $\mc{S} = \{s_1, s_2, \ldots\}$.  We proceed by induction on $|\mc{S}|$.  Suppose $|\mc{S}| = 2$, and let $p = P(s_1)$ and $q = Q(s_1)$.  The entropy function $h_b(x)$ is concave, increasing on $[0,1/2]$ and decreasing on $[1/2,1]$.  Applying Lemma \ref{lemma:concaveLem} to each interval, we obtain the bound:
	\begin{align}
	|h_b(x + \epsilon) - h_b(x)| \le h_b(\epsilon)~.
	\end{align}
Since $H(P) = h_b(p)$ and $H(Q) = h_b(q)$, this proves our result.

Now suppose that the lemma holds for $|\mc{S}| \le m - 1$, and consider the case $|\mc{S}| = m$.  Without loss of generality, let $P(s_m) > 0$ and $Q(s_m) > 0$.  Let $\lambda = (1 - P(s_m))$ and $\mu = (1 - Q(s_m))$ and note that $|\lambda - \mu| < \epsilon$ by assumption.  Define the $(m-1)$ dimensional distributions $P' = \lambda^{-1} (P(s_1), \ldots P(s_{m-1}))$ and $Q' = \lambda^{-1} (Q(s_1), \ldots Q(s_{m-1}))$, so that
	\begin{align*}
	P &= (\lambda P', (1 - \lambda)) \\
	Q &= (\mu Q', (1 - \mu))~.
	\end{align*}
Therefore,
	\begin{align*}
	H(P) &= h_b(\lambda) + \lambda H(P') \\
	H(Q) &= h_b(\mu) + \mu H(Q')~.
	\end{align*}
	
Now we we can expand the difference of the entropies.  Using the fact that $\lambda < 1$, the induction hypothesis on $|H(P') - H(Q')|$ and
$|h_b(\lambda) - h_b(\mu)|$, and the cardinality bound on the entropy $H(Q')$ yields the result:
	\begin{align*}
	|H(P) - H(Q)| &= |\lambda H(P') - \mu H(Q') + h_b(\lambda) - h_b(\mu) | \\
	&\le \lambda |H(P') - H(Q')| + |\lambda - \mu| H(Q') + |h_b(\lambda) - h_b(\mu) | \\
	&\le (m - 2) \cdot h_b(\epsilon) + (m-2) \log(m-2) \cdot \epsilon + \log(m-1) \cdot \epsilon + h_b(\epsilon) \\
	&\le (m - 1) \cdot h_b(\epsilon) + (m-1) \log(m-1) \cdot \epsilon~.
	\end{align*}
\end{proof}

\begin{lemma} \label{lemma:var_bnd_mi} 
Let $W(y | x)$ and $V(y | x)$ be two channels with finite input and output alphabets $\mc{X}$ and $\mc{Y}$.  If
	\begin{align}
		\left| W(y | x) - V(y | x) \right| \le \epsilon \qquad \forall (x,y) \in \mc{X} \times \mc{Y}~,
	\end{align}
then for any input distribution $P$ on $\mc{X}$ we have
	\begin{align}
		\left| I(P,W) - I(P,V) \right| \le 2 (|\mc{Y}| - 1) \cdot h_b(\epsilon) + 2 (|\mc{Y}| - 1) \log(|\mc{Y}| - 1) \cdot \epsilon~,
	\end{align}
where $h_b(\cdot)$ is the binary entropy function.
\end{lemma}

\begin{proof}
We simply apply Lemma \ref{lemma:var_bnd_ent} twice.  Let $Q_W$ and $Q_V$ be the marginal distributions on $\mc{Y}$ under channels $W$ and $V$ respectively.  Then
	\begin{align*}
	|Q_W(y) - Q_V(y)| \le \sum_{x} P(x) |W(y | x) - V(y | x)| \le \epsilon~.
	\end{align*}
Now we can break apart the mutual information and use Lemma \ref{lemma:var_bnd_ent} on each term:
	\begin{align*}
	|I(P,W) - I(P,V)| &\le | H(Q_W) - H(Q_V)| + \sum_{x} P(x) |H(W(Y | X = x)) - H(V(Y | X = x)) | \\
	&\le 2 (|\mc{Y}| - 1) \cdot h_b(\epsilon) + 2 (|\mc{Y}| - 1) \log(|\mc{Y}| - 1) \cdot \epsilon~.
	\end{align*}
\end{proof}

\subsection{Properties of concatenated fixed composition sets}

Let $\tau(\mbf{x})$ be the type of $\mbf{x}$.  Let $\mbf{T}_n(P) = \{\mbf{x} \in \mc{X}^{n} : \tau(\mbf{x}) = P\}$ be the set of of all length-$n$ vectors of type $P$.  For a vector $\mbf{x}$, let $\mbf{x}_1^m$ be the first $m$ elements of $\mbf{x}$.

\begin{lemma} \label{lemma:cat_set_bd} 
For all finite sets $\mathcal{X}$, and all types $P$ with $p_0 = \min_{x \in \mc{X}} P(x) > 0$, there exists $\eta = \eta(P) < \infty$ such that for sufficiently large $n$, for all $M > 0$:
	\begin{align}
		\frac{|\mbf{T}_n(P)|^{M}}{|\mbf{T}_{Mn}(P)|} \ge \exp( - \eta M \log n )~.
	\end{align}
\end{lemma}
\begin{proof}
We begin with the following \cite[p. 39]{CsiszarKorner:82book} :
	\begin{align*}
	k H(P) - \frac{|\mc{X}|-1}{2} \log(2 \pi k) - \nu_1(P)
	\le 
	\log{|\mbf{T}_k(P)|} 
	\le
	k H(P) - \frac{|\mc{X}|-1}{2} \log(2 \pi k) - \nu_2(P)~,
	\end{align*}
for $0 < \nu_1(P) < \infty$ and $0 < \nu_2(P) < \infty$ since $p_x \ge p_0$ for all $x$.
From this we can take the ratio:
	\begin{align*}
	\log \frac{|\mbf{T}_n(P)|^{M}}{|\mbf{T}_{Mn}(P)|}
	&\ge
	-M \frac{|\mc{X}|-1}{2} \log(2 \pi n) - M \nu_1(P)
	+ \frac{|\mc{X}|-1}{2} \log(2 \pi Mn) + \nu_2(P)
	\end{align*}
For fixed $P$ and sufficiently large $n$, this lower bound is $\Omega( M \log n )$, which establishes the result.
\end{proof}

\bibliographystyle{IEEEtran}
\bibliography{indseq}

\end{document}